\documentclass[a4paper,11pt]{article}

\newcommand{\managepath}{$2(k+1) + 6\kappa$}
\newcommand{\NEWsizeker}{$\OO(k^{12}\log^{10} k)$}
\newcommand{\NEWinitsizeker}{$\OO(k^{22}\log^{12} k)$}
\newcommand{\NEWapproxfactor}{$\OO(\log^2 n)$}
\newcommand{\NEWapproxsolsizek}{$\OO(k\log^2 k)$}
\newcommand{\NEWinitapproxsolsizek}{$\OO(k^3\log^2 k)$}
\newcommand{\NEWinitapproxfactor}{$\OO(k^2\log^2 k)$}

\usepackage{mymacros}
\usepackage{mydebug}

\newcommand{\cdel}{\textsc{CVD}\xspace}

\newcommand{\OO}{{\mathcal O}}

\newcommand{\opt}{{\sf opt}}

\newcommand{\defparproblem}[4]{
\vspace{3mm}
\noindent\fbox{
  \begin{minipage}{.95\textwidth}
  \begin{tabular*}{\textwidth}{@{\extracolsep{\fill}}lr} \textsc{#1}  & {\bf{Parameter:}} #3 \\ \end{tabular*}
  {\bf{Input:}} #2  \\
  {\bf{Question:}} #4
  \end{minipage}
  }
\vspace{2mm}}

\newcommand{\defparproblemOpt}[3]{
\vspace{3mm}
\noindent\fbox{
  \begin{minipage}{.95\textwidth}
  \begin{tabular*}{\textwidth}{@{\extracolsep{\fill}}lr} \textsc{#1}  \\ \end{tabular*}
  {\bf{Input:}} #2  \\
  {\bf{Question:}} #3
  \end{minipage}
  }
\vspace{2mm}}

\usepackage{algorithmic}
\newcommand{\alg}[1]{\mbox{\sf #1}}  

\usepackage{vmargin}
\setmarginsrb{1in}{1in}{1in}{1in}{0pt}{0pt}{0pt}{6mm}

\begin{document}

\title{Feedback Vertex Set Inspired Kernel for Chordal Vertex Deletion\thanks{A preliminary version of this paper appeared in the proceedings of 28th ACM-SIAM Symposium on Discrete Algorithms (SODA 2017). The research leading to these results received funding from the European Research Council under the European Union’s Seventh Framework Programme (FP/2007-2013) / ERC Grant Agreement no.~306992.}}
\author{
Akanksha Agrawal\thanks{
University of Bergen, Bergen, Norway. \texttt{akanksha.agrawal@uib.no}.
} \and 
Daniel Lokshtanov\thanks{
University of Bergen, Bergen, Norway. \texttt{daniello@uib.no}.
} \and 
Pranabendu Misra\thanks{
The Institute of Mathematical Sciences, HBNI, Chennai, India. \texttt{pranabendu@imsc.res.in}.
}
\and 
Saket Saurabh\thanks{
University of Bergen, Bergen, Norway. The Institute of Mathematical Sciences, HBNI, Chennai, India. \texttt{saket@imsc.res.in}.
}
\and 
Meirav Zehavi\thanks{
University of Bergen, Bergen, Norway. \texttt{meirav.zehavi@uib.no}.
}}

\begin{titlepage}
\def\thepage{}
\thispagestyle{empty}
\maketitle
\vspace{-1cm}
\begin{abstract} 
Given a graph $G$ and a parameter $k$, the {\sc Chordal Vertex Deletion (CVD)} problem asks whether there exists a subset $U\subseteq V(G)$ of size at most $k$ that hits all induced cycles of size at least 4. The existence of a polynomial kernel for {\sc CVD} was a well-known open problem in the field of Parameterized Complexity. Recently, Jansen and Pilipczuk resolved this question affirmatively by designing a polynomial kernel for {\sc CVD} of size $\OO(k^{161}\log^{58}k)$, and asked whether one can design a kernel of size $\OO(k^{10})$. While we do not completely resolve this question, we design a significantly smaller kernel of size \NEWsizeker, inspired by the $\OO(k^2)$-size kernel for {\sc Feedback Vertex Set}. Furthermore, we introduce the notion of the independence degree of a vertex, which is our main conceptual contribution.
\end{abstract}
 \end{titlepage}
\newpage 

%
%
%
%


\section{Introduction}

Data reduction techniques are widely applied to deal with computationally hard problems in real world applications. It has been a long-standing challenge to formally express the efficiency and accuracy of these ``pre-processing'' procedures. The framework of parameterized complexity turns out to be particularly suitable for a mathematical analysis of pre-processing heuristics. Formally, in parameterized complexity each problem instance is accompanied by a parameter $k$, and we say that a parameterized problem is {\em fixed-parameter tractable (FPT)} if there is an algorithm that solves the problem in time $f(k)\cdot |I|^{\OO(1)}$, where $|I|$ is the size of the input and $f$ is a computable function of the parameter $k$ alone. 
{\em Kernelization} is the subarea of parameterized complexity that deals with the mathematical 
analysis of pre-processing heuristics. 
A parameterized problem is said to admit a {\em polynomial kernel} if there is a polynomial-time algorithm (the degree of polynomial is independent of the parameter $k$), called a {\em kernelization algorithm}, that reduces the input instance down to an instance whose size is bounded by a polynomial $p(k)$ in $k$, while preserving the answer. This reduced instance is called a {\em $p(k)$-kernel} for the problem.
Observe that if a problem has a kernelization algorithm, then it is also has an FPT algorithm.

Kernelization appears to be an interesting computational approach not only from a theoretical perspective, but also from a practical perspective.
There are many real-world applications where even very simple preprocessing can be surprisingly effective, leading to significant size-reduction of the input. Kernelization is a natural tool for  measuring the quality of existing preprocessing rules proposed for  specific problems as well as for designing new powerful such rules. The most fundamental question in the field of kernelization is: 
\begin{quote}
  Let $\Pi$ be parameterized problem that admits an \FPT\  algorithm. Then, does $\Pi$ admit a polynomial kernel?
\end{quote}
In recent times, the study of kernelization, centred on the above question, has been one of the main areas of research in parameterized complexity, yielding many new important contributions to theory.
These include general results showing that certain classes of parameterized problems have polynomial kernels, and as well as other results that utilize advanced techniques from algebra, matroid theory and topology for data reduction~\cite{AlonGKSY11,H.Bodlaender:2009ng,FominLST10,FominLMS12,GiannopoulouJLS15,Jansen14,KimLPRRSS16,Kratsch12a,KratschW14,KratschW12,PilipczukPSL14,Thomasse10}. The development of a framework for ruling out polynomial kernels under certain complexity-theoretic assumptions~\cite{BDFH08,BodlaenderJK13a,DellM14,FortnowS11} has added a new dimension to the area, and strengthened its connections to classical complexity theory.
We refer to the following surveys~\cite{Kratsch14,LokshtanovMS12} and the corresponding chapters in the books~\cite{paramalgoCFKLMPPS,ParameterizedComplexityBook,FG06,Nie06},
for a detailed introduction to the field of kernelization.

An important class of problems, that has led to the development of many upper bound tools and techniques in kernelization, is the class of {\em parameterized graph deletion problems}. A typical problem of this class is associated with a family of graphs, $\cal F$,  such as edgeless graphs, forests, cluster graphs, chordal graphs, interval graphs, bipartite graphs,  split graphs or planar graphs. The deletion problem corresponding to $\cal F$ is formally stated as follows. 

\defparproblem{{\sc $\cal F$-Vertex (Edge) Deletion}}{An undirected graph $G$ and a non-negative integer $k$.}{$k$}{Does there exist $S\subseteq V(G)$ (or $S\subseteq E(G)$) such that $|S|\leq k$ and $G\setminus S$ is in $\cal F$?}
 
 \noindent 
Graph deletion problems are also among the most basic problems in graph theory and graph algorithms. Most of these problems are NP-complete~\cite{Yannakakis:1978,LewisY80}, and thus they were subject to intense study in various algorithmic paradigms to cope with their intractability~\cite{Fujito98,LundY94,FominLMS12,MOR13}. 
These include, considering a restricted class of inputs, approximation algorithms, parameterized complexity and kernelization.  

Some of the most well known results in kernelization are  polynomial kernels 
for graph deletion problems such as {\sc Feedback Vertex Set}~\cite{Thomasse10}, {\sc Odd Cycle Transversal}~\cite{KratschW14,KratschW12}, {\sc Vertex Cover}~\cite{Abu-KhzamCFLSS04,ChenKJ01}, {\sc Planar-$\cal F$-Deletion}~\cite{FominLMS12}, and {\sc Treedepth-$\eta$-Deletion}~\cite{GiannopoulouJLS15}. A common thread among all these problems, with the exception of {\sc Odd Cycle Transversal}, is that the corresponding family $\cal F$ can be characterized by a finite set of forbidden minors that include at least one connected planar graph. It is known that, if $\cal F$ can be characterized by a finite set of forbidden induced subgraphs, then the corresponding {\sc $\cal F$-Vertex Deletion} problem immediately admits an \FPT\ algorithm as well as polynomial sized kernel because of its connection to {\sc $d$-Hitting Set}~ \cite{Abu-Khzam10}. However, if $\cal F$ is characterized by an {\em infinite} set of forbidden induced subgraphs, which is the case when $\cal F$ is the class of chordal graphs, chordal bipartite graphs, interval graphs, proper interval graphs and permutation graphs, our understanding of these problems in the realm of parameterized complexity and kernelization, is still at a nascent stage. While {\sc Chordal Vertex Deletion (CVD)} was known to be \FPT\ for some time~\cite{Marx10,CaoM16}, the parameterized complexity of {\sc Interval Vertex Deletion} was settled only recently~\cite{CaoM15,Cao16}. The parameterized complexity of {\sc Permutation Vertex Deletion} and {\sc Chordal Bipartite Vertex Deletion} is still unknown. Coming to the question of polynomial kernels for these problems, the situation is even more grim. Until recently, the only known result was a polynomial kernel for {\sc Proper Interval Vertex Deletion}: Fomin et al.~\cite{FominSV13} obtained a $\OO(k^{53})$ sized polynomial kernel for {\sc Proper Interval Vertex Deletion}, which has recently been improved to $\OO(k^{4})$~\cite{ke2016unit}. A dearth of further results in this area has led to the questions of kernelization complexity of 
 {\sc Chordal Vertex Deletion} and {\sc Interval Vertex Deletion} becoming prominent open problems~\cite{CaoM16,cyganopen,FominSV13,HeggernesHJKV13,Marx10}.

Jansen and Pilipzuck~\cite{JansenPili2016} recently resolved one of these open questions. They showed that \cdel admits a polynomial kernel of size $\OO(k^{161}\log^{58} k)$, and further posed an open question: 
\begin{quote}
Does \cdel admit a kernel of size  $\OO(k^{10})$? 
\end{quote}
While we do not completely resolve this question, we design a significantly smaller kernel. In particular, we obtain the following result. 
\begin{theorem}\label{thm:chordal}
 	\cdel admits a polynomial kernel of size \NEWsizeker.
 \end{theorem}
 
 \paragraph{Our Methods.} Our result is  inspired by the $\OO(k^2)$-size kernel for 
 {\sc Feedback Vertex Set (FVS)} (checking whether there exists a $k$ sized vertex subset that intersects all cycles) , designed by Thomass\'{e}~\cite{Thomasse10}. The kernel for 
 {\sc FVS} consists of the two following steps. 
 \begin{enumerate} 
 \item Reduce the maximum degree of the graph by using matching based tools (in particular expansion lemma). That is, upper bound the maximum degree $\Delta$ of the graph by $\OO(k)$. 
 \item When the graph has maximum degree $\Delta$, one can show that if a graph has minimum  degree at least $3$ (which can be easily achieved for  {\sc FVS}) then any minimum feedback vertex set has size $\OO(n/\Delta)$. This together with an upper bound on $\Delta$ implies 
 $\OO(k^2)$-size kernel. 
 \end{enumerate}

Let us now look at the \cdel problem. Here, our objective is to check whether there exists a $k$-sized vertex subset $S$ such that $G\setminus S$ is a chordal graph.
However, what is a chordal graph? A graph is chordal if it does not contain any induced cycle of length at least $4$.  That is, every cycle of length at least $4$ has a chord.
Thus, {\sc FVS} is about intersecting all cycles and \cdel\ is about intersecting all chordless cycles. Unfortunately, this apparent similarity stops here!
Nevertheless, we are still able to exploit ideas used in the $\OO(k^2)$-size kernel for {\sc FVS}. 
Towards this, we define the notion of {\em independence degree} of vertices and graphs. Roughly speaking,  the independence degree of a vertex $v$ is the size of a maximum independent set in its neighborhood ($G[N(v)]$). The study of this notion is our main conceptual contribution.

 As a first step we bound the independence degree of every vertex by $k^{\OO(1)}$ -- this is similar to the first step of the kernel for 
{\sc FVS}. Once we have bounded the independence degree of a graph, we obtain an approximate solution $M$ (also called modulator) and analyze the graph $G\setminus M$. The bound on the independence degree 
immediately implies that the number of leaves, vertices of degree at least three and the number of maximal degree two paths in the clique forest of $G\setminus M$ is bounded by $k^{\OO(1)}$. Then, using ideas similar to those used by Marx~\cite{Marx10} to bound the size of a maximal clique while designing the first algorithm for \cdel, we reduce the size of each maximal clique in $G\setminus M$ to $k^{\OO(1)}$. Finally, we use structural analysis to bound the size of each maximal degree two path, which includes the design of a reduction rule that computes a family of minimum cuts, and thus we obtain our final kernel. We believe that the notion of independent degree is likely to be useful for designing algorithms for other graph modification problems. Not only does it lead to a significant improvement, once it is bounded, it greatly simplifies the analysis of the resulting instance.

Finally, after we obtain a kernel, we show that by reruning our entire kernelization procedure once, we can actually reduce the size of the kernel. Since we rely on an \NEWapproxfactor-approximation algorithm, when we call our kernelization procedure for the first time, we work with an approximate solution of size \NEWinitapproxsolsizek; indeed, it can be assumed that $\log n < k\log k$, else the $2^{\OO(k\log k)}\cdot n^{\OO(1)}$-time algorithm for \cdel by Cao and Marx \cite{CaoM16} solves the input instance in polynomial time. However, once we have our first kernel, it holds that $n=$\NEWinitsizeker. At this point, if we reuse the approximation algorithm, we obtain an approximate solution of size \NEWapproxsolsizek. The size of our kernel depends on the size of the approximate solution; in particular, having an approximate solution of size \NEWapproxsolsizek\ allows us to obtain a kernel of size \NEWsizeker.

\section{Preliminaries}
 
For a positive  integer $k$, we use  $[k]$  as a shorthand for  $\{1,2,\ldots,k\}$. Given a function $f: A\rightarrow B$ and a subset $A'\subseteq A$, we let $f|_{A'}$ denote the function $f$ restricted to the domain $A'$.

\bigskip
{\noindent\bf Parameterized Complexity.} In Parameterized Complexity each problem instance is accompanied by a parameter $k$. A central notion in this field is the one of {\em fixed-parameter tractability (FPT)}. This means, for a given instance $(I, k)$, solvability in time $f(k)|I|^{\OO(1)}$ where $f$ is some computable function of $k$. A parameterized problem is said to admit a {\em polynomial kernel} if there is a polynomial-time algorithm (the degree of polynomial is independent of the parameter $k$), called a {\em kernelization algorithm}, that reduces the input instance down to an equivalent instance whose size is bounded by a polynomial $p(k)$ in $k$. Here, two instances are equivalent if one of them is a yes-instance if and only if the other one is a yes-instance. The reduced instance is called a {\em $p(k)$-kernel} for the problem. For a detailed introduction to the field of kernelization, we refer to the following surveys~\cite{Kratsch14,LokshtanovMS12} and the corresponding chapters in the books~\cite{paramalgoCFKLMPPS,ParameterizedComplexityBook,FG06,Nie06}.

Kernelization algorithms often rely on the design of {\em reduction rules}. The rules are numbered, and each rule consists of a condition and an action. We always apply the first rule whose condition is true. Given a problem instance $(I, k)$, the rule computes (in polynomial time) an instance $(I',k')$ of the same problem where $k'\leq k$. Typically, $|I'|<|I|$, where if this is not the case, it should be argued why the rule can be applied only polynomially many times. We say that the rule {\em safe} if the instances $(I,k)$ and $(I',k')$ are equivalent.

\bigskip
{\noindent\bf Graphs.} Given a graph $G$, we let $V(G)$ and $E(G)$ denote its vertex-set and edge-set, respectively. In this paper, we only consider undirected graphs. We let $n=|V(G)|$ denote the number of vertices in the graph $G$, where $G$ will be clear from context. 
The \emph{open neighborhood}, or simply the \emph{neighborhood}, of a vertex $v\in V(G)$ is defined as $N_G(v) = \{w \mid \{v,w\} \in E(G) \}$. The \emph{closed neighborhood} of $v$ is defined as $N_G[v] = N_G(v) \cup \{ v \}$. 
The \emph{degree} of $v$ is defined as $d_G(v) = |N_G(v)|$.
We can extend the definition of neighborhood of a vertex to a set of vertices as follows.
Given a subset $U \subseteq V(G)$, $N_G(U) = \bigcup_{u\in U} N_G(u)$ and $N_G[U] = \bigcup_{u\in U} N_G[u]$. The \emph{induced subgraph} $G[U]$ is the graph with vertex-set $U$ and edge-set $\{\{u,u'\}~|~u,u'\in U, \text{ and } \{u,u'\} \in E(G)\}$.
Moreover, we define $G \setminus U$ as the induced subgraph $G[V(G) \setminus U]$. We omit subscripts when the graph $G$ is clear from context.
An \emph{independent set} in $G$ is a set of vertices such that there is no edge between any pair of vertices in this set.
The \emph{independence number} of $G$, denoted by $\alpha(G)$, is defined as the cardinality of the largest independent set in $G$.
A \emph{clique} in $G$ is a set of vertices such that there is an edge between every pair of vertices in this set.

A \emph{path} $P$ in $G$ is a subgraph of $G$ where
$V(P) = \{ x_1, x_2, \ldots, x_\ell \} \subseteq V(G)$ and $E(P) = \{ \{x_1,x_2\}, \{x_2, x_3\}, \ldots, \{x_{\ell-1},x_\ell\}\}\subseteq E(G)$ for some $\ell\in[n]$.
The vertices $x_1$ and $x_\ell$ are called \emph{endpoints} of the path $P$ and the remaining vertices in $V(P)$ are called \emph{internal vertices} of $P$.
We also say that $P$ is a path between $x_1$ and $x_\ell$.
A \emph{cycle} $C$ in $G$ is a subgraph of $G$ where
$V(C) = \{ x_1, x_2, \ldots, x_\ell \} \subseteq V(G)$ and $E(C) = \{ \{x_1,x_2\}, \{x_2, x_3\}, \ldots, \{x_{\ell-1},x_\ell\}, \{x_\ell, x_1\}\} \subseteq E(G)$, i.e., it is a path with an additional edge between $x_1$ and $x_\ell$.
Let $P$ be a path in the graph $G$ on at least three vertices.
We say that $\{u,v\} \in E(G)$ is a \emph{chord} of $P$ if $u,v \in V(P)$ but $\{u,v\} \notin E(P)$.
Similarly, for a cycle $C$ on at least four vertices, $\{u,v\} \in E(G)$ is a chord of $C$ if $u,v \in V(C)$ but $\{u,v\} \notin E(C)$.
A path $P$ or cycle $C$ is \emph{chordless} if it has no chords.

The graph $G$ is \emph{connected} if there is a path between every pair of vertices,
otherwise $G$ is \emph{disconnected}. A connected graph without any cycles is a \emph{tree}, and a collection of trees is a \emph{forest}. A maximal connected subgraph of $G$ is called a \emph{connected component} of~$G$. 

\bigskip
{\noindent\bf Forest Decompositions.} 
A \emph{forest decomposition} of a graph $G$ is a pair $(F,\beta)$ where $F$ is forest,
and $\beta:V(T) \rightarrow 2^{V(G)}$ 
is a function that satisfies the following,
\begin{enumerate}[(i)]
    \item $\bigcup_{v \in V(F)} \beta(v) = V(G)$,
    \item for any edge $\{v,u\} \in E(G)$ there is a node $w \in V(F)$ such that $v,u \in \beta(w)$, 
    \item and for any $v \in V(G)$, the collection of nodes $T_v = \{ u \in V(F) \mid v \in \beta(u)\}$ is a subtree~of~$F$.
\end{enumerate}
For $v \in V(F)$, we call $\beta(v)$ the \emph{bag} of $v$, and for the sake of clarity of presentation, we sometimes use $v$ and $\beta(v)$ interchangeably. We refer to the vertices in $V(F)$ as nodes.
A \emph{tree decomposition} is a forest decomposition where $F$ is a tree.

\bigskip
{\noindent\bf Chordal Graphs.} A graph $G$ is a \emph{chordal graph} if it has no chordless cycle as an induced subgraph,
i.e., every cycle of length at least four has a chord. A \emph{clique forest} of $G$ is a forest decomposition of $G$ where every bag is a maximal clique. The following lemma shows that the class of chordal graphs is exactly the class of graphs which have a clique~forest.
\begin{lemma}[Theorem 4.8,~\cite{Golumbic80}]\label{lem:cliqueForest}
A graph $G$ is a chordal graph if and only if $G$ has a clique forest.
\end{lemma}

Given a subset $U\subseteq V(G)$, we say that $U$ {\em hits} a chordless cycle $C$ in $G$ if $U\cap V(C)\neq\emptyset$.
Observe that if $U$ \emph{hits} every chordless cycle of $G$, then $G\setminus U$ is a chordal graph. Given a $v\in V(G)$, we say that a vertex-set $B\subseteq V(G)\setminus\{v\}$ is a {\em $v$-blocker} if $B$ hits every chordless cycle in $G$. Observe that the set $B$ must {\em not} contain the vertex $v$. 

The {\sc Chordal Vertex Deletion (CVD)} problem is defined as follows.

\defparproblem{{\sc Chordal Vertex Deletion (CVD)}}{An undirected graph $G$ and a non-negative integer $k$.}{$k$}{Does there exist a subset $S\subseteq V(G)$ such that $|S|\leq k$ and $G\setminus S$ is a chordal graph?}

For purposes of approximation, we also formulate this problem as an optimization problem.

\defparproblemOpt{{\sc Chordal Vertex Deletion (CVD)}}{An undirected graph $G$.}{What is the minimum cardinality of a subset $S\subseteq V(G)$ such that $G\setminus S$ is a chordal graph?}

\bigskip
{\noindent\bf The Expansion Lemma.}
Let $c$ be a positive integer.
A \emph{$c$-star} is a graph on $c+1$ vertices where one vertex, called the center, has degree $c$, and all other vertices are adjacent to the center and have degree one.
A \emph{bipartite graph} is a graph whose vertex-set can be partitioned into two independent sets.
Such a partition the vertex-set is called a \emph{bipartition} of the graph.
Let $G$ be a bipartite graph with bipartition $(A,B)$.
A subset of edges $M \subseteq E(G)$ is called {\em $c$-expansion of $A$ into $B$} if 
\begin{enumerate}[(i)]
    \item every vertex of $A$ is incident to exactly $c$ edges of $M$,
    \item and $M$ saturates exactly $c|A|$ vertices in $B$.
\end{enumerate}
Note that a $c$-expansion saturates all vertices of $A$, and for each $u \in A$ the set of edges in $M$ incident on $u$ form a $c$-star. 
The following lemma allows us to compute a $c$-expansion in a bipartite graph.
It captures a certain property of neighborhood sets which is very useful for designing kernelization algorithms.
\begin{lemma}[\cite{Thomasse10,paramalgoCFKLMPPS}]\label{lem:expansion}
Let $G$ be a bipartite graph with bipartition $(A,B)$ such that there are no isolated vertices in $B$. Let $c$ be a positive integer such that $|B| \geq c |A|$. 
Then, there are non-empty subsets $X \subseteq A$ and $Y \subseteq B$ such that 
\begin{itemize}
    \item there is a $c$-expansion from $X$ into $Y$,
    \item and there is no vertex in $Y$ that has a neighbor in $A \setminus X$, i.e. $N_G(Y) = X$.
\end{itemize}
Further, the sets $X$ and $Y$ can be computed in polynomial time.
\end{lemma}


\section{Kernelization}\label{sec:chord}

In this section we prove Theorem \ref{thm:chordal}.
First, in Section~\ref{sec:chordApprox}, we briefly state results relating to approximate solutions for \cdel, which will be relevant to following subsections. In Section \ref{sec:irrelevant} we address annotations that will be added to the input instance.
Next, in Section~\ref{sec:chordIndependent}, we introduce the notion of the {\em independent degree} of a vertex, which lies at the heart of the design of our kernelization algorithm. We carefully examine the independent degrees of vertices in our graphs, and show how these degrees can be bounded by a small polynomial in $k$. In Section~\ref{sec:chordCliqueForest}, we consider the clique forest of the graph obtained by removing (from the input graph) the vertices of an approximate solution. In particular, we demonstrate the usefulness of our notion of an independent degree of a vertex -- having bounded the independent degree of each vertex, we show that the number of leaves in the clique forest can be bounded in a simple and elegant manner. We also efficiently bound the size of a maximal clique. Then, in Section~\ref{sec:ChordPaths}, we turn to bound the length of degree-2 paths in the clique forest. This subsection is quite technical, and its proofs are based on insights into the structure of chordal graphs and their clique forests. In particular, we use a reduction rule which computes a collection of minimum cuts rather than one minimum cut which overall allows us to capture the complexity of a degree-2 path using only few vertices. In Section \ref{sec:unmark}, we remove annotations introduced in preceding sections. Next, in Section~\ref{sec:chordProof}, we bound the size of our kernel.
Finally, in Section \ref{sec:better}, we show that an alternating application of approximation and kernelization can improve the performance of our kernelization algorithm.

\subsection{Approximation}\label{sec:chordApprox}





Observe that it can be assumed that $\log n < k\log k$, else the $2^{\OO(k\log k)}\cdot n^{\OO(1)}$-time algorithm for \cdel by Cao and Marx \cite{CaoM16} solves the input instance in polynomial time. Thus, since \cdel admits an \NEWapproxfactor-factor approximation algorithm \cite{manuscript}, we obtain the following result.

\begin{lemma}\label{cor:newApprox}
\cdel admits an \NEWinitapproxfactor-factor approximation algorithm.
\end{lemma}

Throughout Section \ref{sec:chord}, we let \alg{APPROX} denote a polynomial-time algorithm for \cdel that returns approximate solutions of size $f(\opt)$ for some function $f$. Initially, it will denote the algorithm given by Lemma \ref{cor:newApprox}.
 Given an instance of \cdel, we say that an approximate solution $D$ is {\em redundant} if for every vertex $v\in D$, $D\setminus\{v\}$ is also an approximate solution, that is, $G\setminus (D\setminus\{v\})$ is a chordal graph. Jansen and Pilipczuk \cite{JansenPili2016} showed that given an approximate solution of size $g(k)$ for some function $g$, one can find (in polynomial time) either a vertex contained in every solution of size at most $k$ or a redundant approximate solution of size $\OO(k\cdot g(k))$. Thus, we have the following result.

\begin{corollary}[{\cite{JansenPili2016}}]
\label{cor:approxSpecial}
Given an instance of \cdel, one can find (in polynomial time) either a vertex contained in every solution of size at most $k$ or a redundant approximate solution of size $\OO(k\cdot f(k))$.
\end{corollary}

Next, we fix an instance $(G,k)$ of \cdel. By relying on \alg{APPROX} and Corollary~\ref{cor:approxSpecial}, we may assume that we have a vertex-set $\widetilde{D}\subseteq V(G)$  that is an approximate solution of size $f(k)$ and a vertex-set $D\subseteq V(G)$ that is a redundant approximate solution of size $c\cdot k\cdot f(k)$ for some constant $c$ independent of the input.

In the following subsection, we will also need to strengthen approximate solutions to be $v$-blockers for some vertices $v\in V(G)$. To this end, we will rely on the following result.

\begin{lemma}\label{lem:approxBlocker}
Given a vertex $v\in V(G)$, one can find (in polynomial time) either a vertex contained in every solution of size at most $k$ or an approximate solution of size $f(k)$ that is a $v$-blocker.
\end{lemma}

\begin{proof}
Fix a vertex $v=v_0\in V(G)$. We define the graph $G'$ by setting $V(G')=V(G)\cup\{v_1,v_2,\ldots,v_{f(k)}\}$, where $v_1,v_2,\ldots,$ $v_{f(k)}$ are new vertices, and $E(G')=E(G)\cup\{\{u,v_i\}:\{u,v\}\in E(G), i\in [f(k)]\} \cup \{(v_i, v_j) : i \neq j\in \{0\}\cup [f(k)]\}$. In other words, $G'$ is the graph $G$ to which we add $f(k)$ copies of the vertex $v$, which form a clique amongst themselves and~$v$.

We call the algorithm \alg{APPROX} to obtain an approximate solution $S$. Since $G'[\{v_0,v_1,\ldots,$ $v_{f(k)}\}]$ is a clique, any chordless cycle in $G'$ contains at most one vertex from $\{v_0,v_1,\ldots,v_{f(k)}\}$. In particular, for any chordless cycle $C$ in $G'$, by replacing $v_i$ by $v_0$ (in case $v_i$ belongs to $C$), we obtain a chordless cycle in $G$. Therefore, if the instance $(G,k)$ admits a solution of size at most $k$ that does not contain $v$, then this solution is also a solution for the instance $(G',k)$, which implies that the size of $S$ should be at most $f(k)$. Thus, we can next assume that $|S|\leq f(k)$, else we conclude that the vertex $v$ is contained in every solution of size at most $k$. Since the vertices $v_0,v_1,v_2,\ldots,v_{f(k)}$ have the same neighbor-set, if $\{v_0,v_1,v_2,\ldots,v_{f(k)}\}\cap S\neq\emptyset$, we can assume that $\{v_0,v_1,v_2,\ldots,v_{f(k)}\}\subseteq S$, since otherwise we can take $S\setminus\{v_0,v_1,v_2,\ldots,v_{f(k)}\}$ as our approximate solution $S$. Since we assume that $|S|\leq f(k)$, we deduce that $\{v_0,v_1,v_2,\ldots,v_{f(k)}\}\cap S=\emptyset$. In particular, we have that $v\notin S$ and therefore $S$ is a $v$-blocker.
\end{proof}

In light of Lemma \ref{lem:approxBlocker}, we may next assume that for every vertex $v\in V(G)$, we have a vertex-set $B_v$ that is both a $v$-blocker and an approximate solution of size $f(k)$.

\subsection{Irrelevant and Mandatory Edges}\label{sec:irrelevant}
During the execution of our kernelization algorithm, we mark some edges in $E(G)$ as {\em irrelevant edges}. At the beginning of its execution, all of the edges in $E(G)$ are assumed to be relevant edges. When we mark an edge as an irrelevant edge, we prove that any solution that hits all of the chordless cycles in $G$ that contain only relevant edges also hits all of the chordless cycles in $G$ that contain the irrelevant edge. In other words, we prove that we can safely ignore chordless cycles that contain at least one irrelevant edge. Observe that we cannot simply remove irrelevant edges from $E(G)$ since this operation may introduce new chordless cycles in $G$.
Instead we maintain a set $E_I$, which contains the edges marked as irrelevant.

We also mark some edges in $E(G)$ as {\em mandatory edges}. 
We will ensure that at least one endpoint of a mandatory edge is present in any solution of size at most $k$.
We let $E_M$ denote the set of mandatory edges.

In some situations, we identify a pair of non-adjacent vertices such that any solution of size at most $k$ must contain at least one of them.
Then, we add an edge between the vertices of the pair, and mark this edge as a mandatory edge. The correctness of this operation follows from the observation that any chordless cycle affected by the addition of this edge contains both vertices of the pair, and since the edge is marked as a mandatory edge, we ensure this chordless cycle will be hit although it may no longer exist.
Formally, we have the following reduction rule.
\begin{reduction}\label{red:req pair}
    Given two non-adjacent vertices in $G$, $v$ and $u$,
    such that at least one of them belongs to any solution 
    of size at most $k$, insert the edge $\{v,u\}$ into both $E(G)$ and~$E_M$.
\end{reduction}

Hence from now onwards our instance is of the form $(G,k,E_I,E_M)$,
and during the execution of our kernelization algorithm, we will update the sets $E_I$ and $E_M$.
In Section \ref{sec:unmark}, we show that we can unmark the edges in $E_I\cup E_M$,
obtaining an ordinary instance of \cdel. For the sake of simplicity, when $E_I$ and $E_M$ are clear from context, we omit them.

\bigskip
{\noindent\bf The Number of Mandatory Edges.}
If a vertex $v$ is incident to at least $k+1$ mandatory edges,
it must belong any solution of size at most $k$.
Therefore, we may safely apply the following reduction~rule.
\begin{reduction}\label{rule:incidentMandatory}
    If there exists a vertex $v$ incident to at least $k+1$ mandatory edges, remove $v$ from $G$ and decrement $k$ by 1.
\end{reduction}
After exhaustively applying the above reduction rule we have the following lemma.
\begin{lemma}
    If $|E_M|>k^2$
    then the input instance is a no-instance.
\end{lemma}
\begin{proof}
    After the exhaustive application of Reduction Rule \ref{rule:incidentMandatory}, any vertex incident to a mandatory edge
    is incident to at most $k$ such edges. Therefore, any set of at most $k$ vertices from $V(G)$ may intersect at most $k^2$ mandatory edge. Since every solution of size at most $k$ must intersect every mandatory edge, we deduce that if $|E_M|>k^2$, the input instance is a no-instance.
\end{proof}

Thus, we will next assume that $|E_M|\leq k^2$ (else Reduction Rule \ref{rule:incidentMandatory} applies).
Moreover, we let $D'$ denote the set $D\cup\widetilde{D}$ to which we add every vertex that is an endpoint of a vertex in $E_M$ (recall that $\widetilde{D}$ is our approximate solution of size at most $f(k)$ and $D$ is 
our redundant approximate solution of size $\OO(k\cdot f(k))$). Observe that by adding vertices to a redundant approximate solution, it remains a redundant approximate solution, and therefore $D'$ or any other superset of $D$ is such a solution.


\subsection{Independent Degree}\label{sec:chordIndependent}



Given a vertex $v\in V(G)$, we use the notation $N^R_G(v)$ to refer to the set of each 
vertex $u\in N_G(v)$ such that $\{v,u\}$ does not belong to $E_I\cup E_M$. We remark that in this subsection we identify and mark some edges as irrelevant edges.

\bigskip
{\noindent\bf Independent Degrees.} We start by introducing the notion of the independent degree of vertices and graphs.

\begin{definition}\label{def:independentDeg}
Given a vertex $v\in V(G)$, the {\em independent degree of $v$}, denoted by $d^I_G(v)$, is the size of a maximum independent set in the graph $G[N^R_G(v)]$. The {\em independent degree of $G$}, denoted by $\Delta^I_G$, is the maximum independent degree of a vertex in $V(G)$.
\end{definition}

Fix $\Delta = (k+3)f(k)$.
The objective of this subsection is to investigate the notion of an independent degree, ultimately proving the following result.

\begin{lemma}\label{lem:boundIndDeg}
One can construct (in polynomial time) an instance $(G',k', E_I', E_R')$ of \cdel that is equivalent to the input instance $(G,k,E_I,E_R)$ and such that both $k'\leq k$ and $\Delta^I_{G'}\leq\Delta$.
\end{lemma}

To this end, we may assume that we are given a vertex $v\in V(G)$ such that $d^I_G(v)>\Delta$. We say that an instance $(G',k',E_I',E_M')$ of \cdel is {\em better} than the input instance $(G,k,E_I,E_M)$ if $k'\leq k$, 
$V(G') = V(G)$, $E_I\subseteq E_I'$, $E_M\subseteq E_M'$
$d^I_{G'}(v)\leq \Delta$ and for all $u\in V(G')$, $d^I_{G'}(u)\leq d^I_G(u)$. To prove the correctness of Lemma \ref{lem:boundIndDeg}, it is sufficient to prove the correctness of the following lemma.

\begin{lemma}\label{lem:boundIndDegV}
We can construct (in polynomial time) an instance $(G',k', E_I', E_M')$ of \cdel that is better than the input instance $(G,k,E_I,E_M)$.
\end{lemma}

Indeed, to prove Lemma \ref{lem:boundIndDeg}, one can repeatedly apply the operation given by Lemma \ref{lem:boundIndDegV} in the context of every vertex $u\in V(G)$ such that $d^I_G(u)>\Delta$. We start with a simple result concerning independent degrees.

\begin{lemma}\label{lem:findIndSet}
Let $u\in V(G)$ be a vertex such that $d^I_G(u)\geq |B_u|$. Then, one can find (in polynomial time) an independent set in $G[N^R_G(u)\setminus B_u]$ of size at least $d^I_G(u)-|B_u|$.
\end{lemma}

\begin{proof}
Since $B_u$ is a solution, $G\setminus B_u$ is a chordal graph. In particular, $G[N^R_G(u)\setminus B_u]$ is a chordal graph. Since chordal graphs are perfect, {\sc Maximum Independent Set} in chrodal graphs is solvable in polynomial time \cite{Golumbic80}. This means that we can find a maximum independent set in $G[N^R_G(u)\setminus B_u]$ in polynomial time. Since the size of a maximum independent set in $G[N^R_G(u)]$ is at least $d^I_G(u)$, the size of a maximum independent set in $G[N^R_G(u)\setminus B_u]$ is at least $d^I_G(u)-|B_u|$, which concludes the correctness of the lemma.
\end{proof}

Recall that for any vertex $u\in V(G)$, $|B_u|\leq f(k)$.
Thus, we have the following corollary.

\begin{corollary}\label{cor:simpleBoundIndDeg}
One can find (in polynomial time) an independent set in $G[N^R_G(v)\setminus B_v]$ of size at least $\Delta-f(k)$.
\end{corollary}

We let $I$ denote the independent set given by Corollary \ref{cor:simpleBoundIndDeg}.

\bigskip
{\noindent\bf Independent Components.} Let $X=N_G(v)\setminus (B_v\cup I)$ denote the neighbor-set of $v$ from which we remove the vertices of the $v$-blocker $B_v$ and of the independent set $I$. We also let $H=G\setminus (\{v\}\cup B_v\cup X)$ denote the graph obtained by removing (from $G$) the vertex $v$, the $v$-blocker $B_v$ and any neighbor of $v$ that does not belong to the independent set $I$. We define the set independent components of $v$ as the set of each connected component of $H$ that contains at least one vertex from $I$, and denote this set by ${\cal A}$.

For the set $\cal A$, we prove the following lemmata.

\begin{lemma}\label{lem:connectedComp1}
Each connected component $A\in{\cal A}$ contains exactly one vertex from $I$ and no other vertex from $N_G(v)$.
\end{lemma}

\begin{proof}
The graph $H$ does not contain any vertex from $N_G(v)\setminus I$, and therefore we only need to prove the first part of the statement of the lemma. Fix a connected component $A\in{\cal A}$. By the definition of ${\cal A}$, it is only necessary to prove that $A$ cannot contain (at least) two vertices from $I$. Suppose, by way of contradiction, that it contains two such vertices, $u$ and $w$. Let $P$ denote the shortest path in $A$ that connects $u$ and $w$. Since $I$ is an independent set, this path contains at least two edges. Therefore, together with the vertex $v$, the path $P$ forms a chordless cycle. However, this chordless cycle contain no vertex from $B_v$, which contradicts the fact that $B_v$ is an approximate solution.
\end{proof}

\begin{figure}[t]\centering
\frame{\includegraphics[scale=0.7]{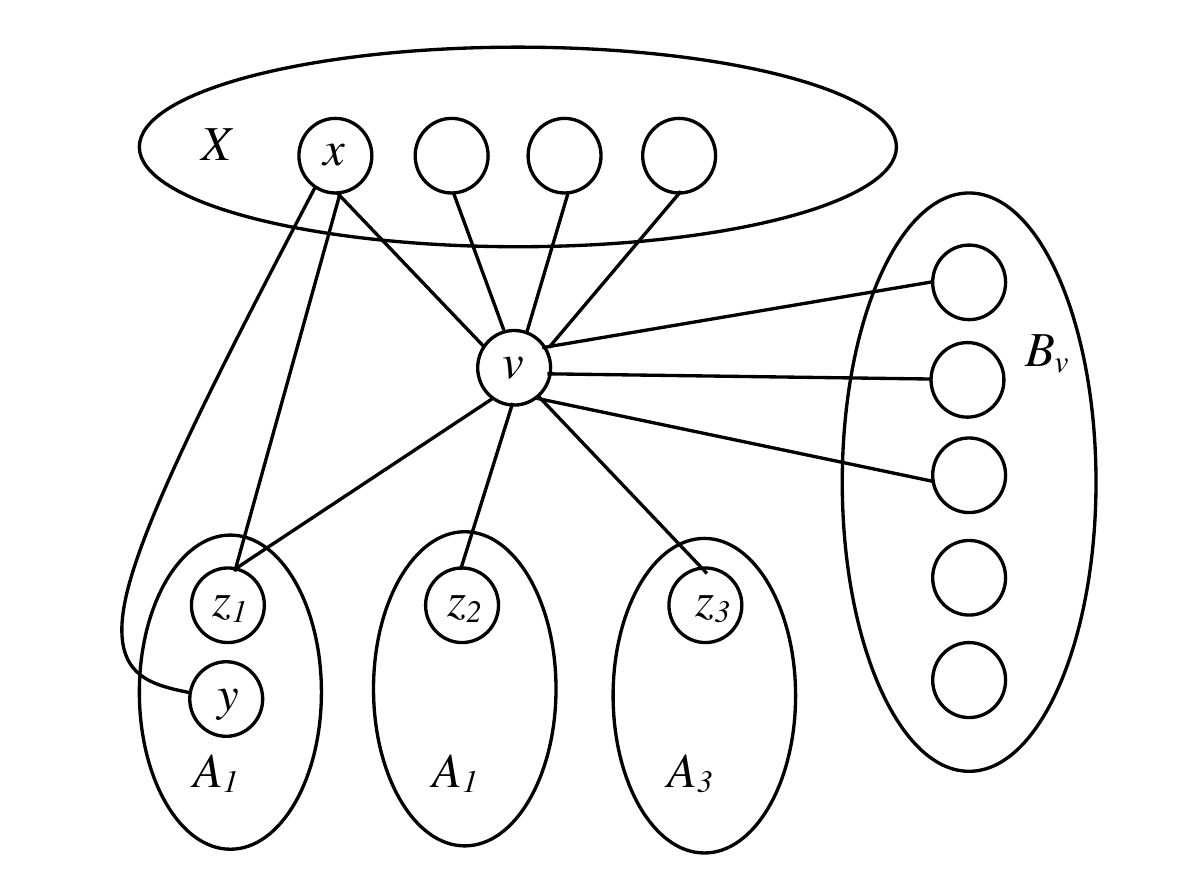}}
\caption{The relations between $v$, $B_v$, $X$ and ${\cal A}$.}\label{fig:lemma24}
\end{figure}

By Lemma \ref{lem:connectedComp1}, for each connected component $A\in{\cal A}$ we can let $z(A)\in I$ denote the unique neighbor of $v$ in $A$. 
In fact, Corollary \ref{cor:simpleBoundIndDeg} and Lemma \ref{lem:connectedComp1} imply that $\Delta-f(k)\leq |{\cal A}|$.

\begin{lemma}\label{lem:connectedComp2}
Every vertex $x\in X$ that is adjacent (in $G$) to some vertex $y\in V(A)$, where $A\in{\cal A}$, is also adjacent (in $G$) to the vertex $z(A)$.
\end{lemma}

\begin{proof}
Let $x\in X$ be a vertex that is adjacent (in $G$) to some vertex $y\in V(A)$. Suppose, by way of contradiction, that $\{x,z(A)\}\notin E(G)$. Let $P$ denote the shortest path in $A$ that connects $y$ and $z(A)$. This path contains at least two edges. Therefore, together with the vertex $v$, the path $P$ forms a chordless cycle. However, this chordless cycle contain no vertex from $B_v$, which contradicts the fact that $B_v$ is an approximate solution.
\end{proof}

An illustration of the relations between $v$, $B_v$ $X$ and ${\cal A}$ is given in Fig.~\ref{fig:lemma24}.

\bigskip
{\noindent\bf The Bipartite Graph $\widehat{H}$.} To decrease $d^I_G(u)$, we will consider the bipartite graph $\widehat{H}$, which is defined as follows. We define vertex-set of $\widehat{H}$ by $V(\widehat{H})={\cal A}\cup B_v$. In this context, we mean that each connected component $A\in{\cal A}$ is represented by a vertex in $V(\widehat{H})$, and for the sake of simplicity, we use the symbol $A$ also to refer to this vertex. We partition $B_v$ into two sets, $B_c$ and $B_f$, where $B_c$ contains the vertices in $B_v$ that are adjacent (in $G$) to $v$, while $B_f$ contains the remaining vertices in $B_v$. Here, the letters $c$ and $f$ stand for ``close'' and ``far''. Having this partition, we define the edge-set of $\widehat{H}$ as follows. For every vertex $b\in B_c$ and connected component $A\in {\cal A}$ such that $b\in N_G(V(A))\setminus N_G(z(A))$ (i.e., $b$ is a neighbor of some vertex in $A$ but not of the vertex $z(A)$), insert the edge $\{b,A\}$ into $E(\widehat{H})$. Moreover, for every vertex $b\in B_f$ and connected component $A\in {\cal A}$ such that $b\in N_G(V(A))$, insert the edge $\{b,A\}$ into $E(\widehat{H})$.
An illustration of the bipartite graph $\widehat{H}$ is given in Fig.~\ref{fig:bipH}. The motivation behind its definition lies at the following lemma.

\begin{figure}[t]\centering
\frame{\includegraphics[scale=0.75]{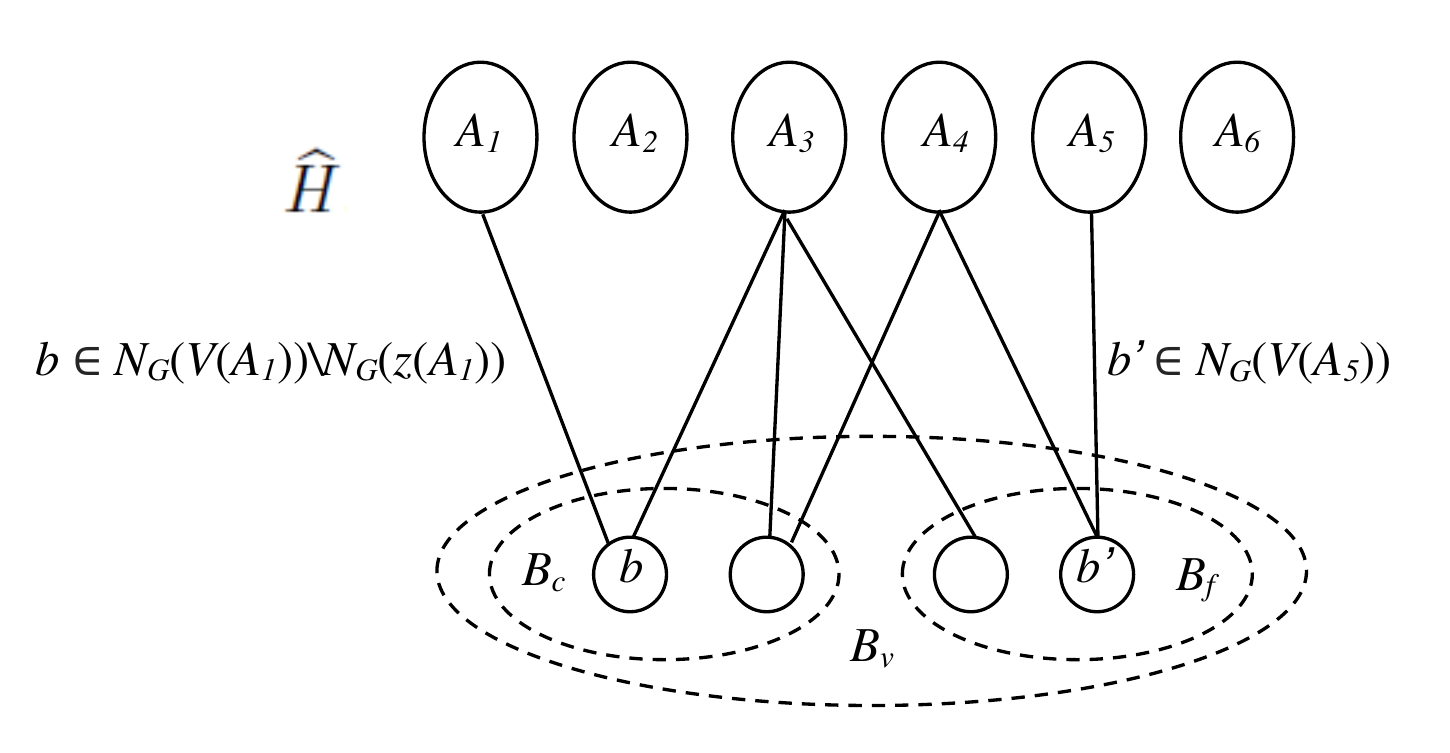}}
\caption{The bipartite graph $\widehat{H}$.}\label{fig:bipH}
\end{figure}

\begin{lemma}\label{lem:H-hat}
The bipartite graph $\widehat{H}$ satisfies the following properties.
\begin{enumerate}
\item Suppose that we are given an edge $\{b,A\}\in E(\widehat{H})$ such that $b\in B_c$ and $A\in{\cal A}$. Then, the graph $G$ has a chordless cycle defined by the edges $\{v,b\}$ and $\{v,z(A)\}$ and a path in $A$.
\item Suppose that we are given edges $\{b,A\},\{b,A'\}\in E(\widehat{H})$ such that $b\in B_f$ and $A,A'\in{\cal A}$. Then, the graph $G$ has a chordless cycle defined by the edges $\{v,z(A)\},\{v,z(A')\}$, the vertex $b$ and paths in $A$ and $A'$.
\end{enumerate}
\end{lemma}

\begin{proof}
To prove the first item, let $y$ be a vertex in $N_G(b)\cap V(A)$, whose existence is guaranteed by the assumption that $\{b,A\}\in E(\widehat{H})$. The desired chordless cycle can be defined by the edges $\{v,b\},\{v,z(A)\},\{b,y\}$ and the edge-set of a shortest path between $z(A)$ and $y$ in $A$ -- in particular, by the definition of ${\cal A}$ and Lemma \ref{lem:connectedComp1}, $v$ is not adjacent to any vertex on this path, excluding $z(A)$, and since $\{b,A\}\in E(\widehat{H})$, $b$ is not adjacent to $z(A)$.

For the second item, choose $y\in N_G(b)\cap V(A)$ and $y'\in N_G(b)\cap V(A')$, whose existence is guaranteed by the assumption that $\{b,A\},\{b,A'\}\in E(\widehat{H})$, such that the length of a shortest path between $z(A)$ ($z(A')$) and $y$ (resp.~$y'$) in $A$ (resp.~$A'$) is minimum. Observe that the cycle defined $\{v,z(A)\},\{v,z(A')\},\{y,b\},\{y',b\}$ and the shortest paths between $z(A)$ and $y$ in $A$ and $z(A')$ and $y'$ in $A'$, respectively, is a chordless cycle --  in particular, by the definition of ${\cal A}$, Lemma \ref{lem:connectedComp1} and since $b,b'\in B_f$, the cycle cannot have a chord containing $v$.
\end{proof}

\bigskip
{\noindent\bf Isolated Vertices in $\widehat{H}$.} We start investigating the bipartite graph $\widehat{H}$ by examining the isolated vertices in ${\cal A}$ that it contains. In this context, we need the two following lemmata.

\begin{lemma}\label{lem:isolated1}
Let $A\in{\cal A}$ be an isolated vertex in $\widehat{H}$, and denote $z=z(A)$. Then, $N_G(V(A))=N_G(z)\setminus V(A)$.
\end{lemma}

\begin{proof}
Since $z\in A$, it is clear that $N_G(z)\setminus V(A)\subseteq N_G(V(A))$. Next, we show that $N_G(V(A))\subseteq N_G(z)\setminus V(A)$. To this end, consider some vertices $y\in V(A)$ and $u\in N_G(y)\setminus V(A)$, and suppose, by way of contradiction, that $u\notin N_G(z)$. Because $A$ is a connected component of $H$, it holds that $u\in \{v\}\cup B_v\cup X$. Moreover, because $A$ is isolated in $\widehat{H}$ and yet $u\in N_G(V(A))\setminus N_G(z)$, it further holds that $u\in X$. However, this results in a contradiction to the statement of Lemma \ref{lem:connectedComp2}.
\end{proof}

\begin{lemma}\label{lem:isolated2}
Let $A\in{\cal A}$ be an isolated vertex in $\widehat{H}$, and denote $z=z(A)$. Then, $G$ does not have a chordless cycle that contains the edge $\{v,z\}$.
\end{lemma}

\begin{proof}
Suppose, by way of contradiction, that $G$ has a chordless cycle $C$ that contains the edge $\{v,z\}$. The cycle $C$ must contain an edge $\{y,u\}$ such that $y\in A$ and $u\notin A$. By Lemma \ref{lem:isolated1}, $N_G(V(A))=N_G(z)\setminus A$, and therefore $\{z,u\}\in E(G)$. Hence, since the cycle $C$ is chordless, it holds that $z=y$, which, in turn, implies that $u\notin N_G(v)$. We deduce that $u$ must belong to $B_f$. However, this implies that $\{A,u\}\in\widehat{H}$, contradicting the assumption that $A$ is an isolated vertex in $\widehat{H}$.
\end{proof}

Lemma \ref{lem:isolated2} leads us to the design of the following reduction rule.

\begin{reduction}
If the graph $\widehat{H}$ contains an isolated vertex $A\in{\cal A}$, mark the edge $\{v,z\}$ as irrelevant.
\end{reduction}

After an exhaustive application of this rule, we can assume that $\widehat{H}$ does not contain an isolated vertex $A\in{\cal A}$.

\bigskip
{\noindent\bf Applying the Expansion Lemma.} Next, we would like to apply Lemma \ref{lem:expansion} in the context of the bipartite graph $\widehat{H}$. Since $|{\cal A}|\geq \Delta - f(k)\geq (k+2)\cdot |B_v|$ and we have already ensured that $\widehat{H}$ does not contain any isolated vertex $A\in{\cal A}$, this lemma implies that we can find (in polynomial time) subsets ${\cal A}^*\subseteq{\cal A}$ and $B^*\subseteq B_v$ such that there exists a $(k+2)$-expansion from $B^*$ into ${\cal A}^*$.

The usefulness of $B^*$ is stated in the following lemma.

\begin{lemma}\label{lem:vOrB}
Any solution of size at most $k$ to the input instance that does not contain $v$ contains all of the vertices in $B^*$.
\end{lemma}

\begin{proof}
Suppose that the input instance is a yes-instance, and suppose it has a solution $S$ of size at most $k$ that does not contain $v$. Consider some vertex $b\in B^*$. Let $A_1,A_2,\ldots,A_{k+2}$ be the neighbors of $b$ in $\widehat{H}$ that correspond to our $(k+2)$-expansion. For any choice of $A_i$ and $A_j$, Lemma \ref{lem:H-hat} implies that if there is no chordless cycle defined by $v$, $b$ and a path in $A_i$, then there is a chordless cycle defined by $v$, $b$, a path in $A_i$ and a path in $A_j$. Therefore, if $S$ contains neither $v$ nor $b$, it must contain at least one vertex from each connected component $A_i$ excluding at most one such component. However, there are $k+2$ such components, and since $S$ does not contain $v$, we deduce that it contains $b$. The choice of $b\in B^*$ was arbitrary, and therefore we conclude that $S$ contains all of the vertices in $B^*$.
\end{proof}

\bigskip
{\noindent\bf Decreasing the Independent Degree of $v$.} Armed with Lemma \ref{lem:vOrB}, we can apply the following reduction rule.


\begin{reduction}\label{rule:decIndDeg}
    For each vertex $b\in B^*$, insert the
    edge $\{b,v\}$ into $E(G)$ (if it is not already present),
    and mark $\{b,v\}$ as a mandatory edge.
    Moreover, mark each edge $\{v,z(A)\}$ such that $A\in{\cal A}^*$ as an irrelevant edge.
\end{reduction}

\begin{lemma}
Reduction Rule \ref{rule:decIndDeg} is safe.
\end{lemma}

\begin{proof}
    First we claim that adding a mandatory edge between
    $v$ and some $b \in B^*$ is safe. Let $G'$ denote the graph obtained after applying this operation.
    By Lemma~\ref{lem:vOrB}, at least one vertex among $v$ and $b$ must be present in any solution of size at most $k$.
    Therefore, we may mark the newly added edge $\{v,b\}$ as a mandatory edge. Since this edge is a mandatory edge, any solution of size at most $k$ to $(G',k)$ will hit any new chordless cycle created by adding this edge as well as any chordless cycle $C$ that exists in $G$ even if in $G'[V(C)]$ is not a chordless cycle (since if $G'[V(C)]$ is not a chordless cycle, $\{v,b\}\subseteq V(C)$). By adding the edges one-by-one, arguing that each operation is safe given that the preceding operation was safe, we derive the safeness of the insertion of all of the edges in $\{\{b,v\}: b\in B^*\}$ as well as the marking of these edges as mandatory.
    
    Now, let $G'$ denote the graph obtained after the application of Reduction Rule \ref{rule:decIndDeg}. 
    Let us argue that it is indeed safe to ignore (in $G'$) any chordless cycle containing an edge $\{v,z(A)\}$ such that $A\in{\cal A}^*$. To this end, consider such a chordless cycle $C$ and a solution $S$ of size at most $k$ to $(G',k)$. If $S$ contains $v$, it clearly hits $C$. Otherwise, we have that $S$ contains all of the vertices in $B^*$. Recall that by Lemma \ref{lem:connectedComp1}, $z(A)$ is the only neighbor of $v$ in $A$, and by Lemma \ref{lem:connectedComp2}, $N_G(z(A))\cap X=N_G(V(A))\cap X$. Therefore, since the cycle $C$ is chordless, it must contain a vertex from $B_v$, but then, by Lemma \ref{lem:expansion}, we further deduce that it must contain a vertex from $B^*$. We thus get that $C$ is hit by $S$.
\end{proof}

Reduction Rule \ref{rule:decIndDeg} decreases $|N^R_G(v)|$. Moreover, as long as $d^R_G(v)>\Delta$, we can apply this rule. Thus, after an exhaustive application of this rule, it should hold that $d^R_G(v)\leq\Delta$. Furthermore, this rule neither inserts vertices into $V(G)$ nor unmarks edges, and therefore we conclude that Lemma \ref{lem:boundIndDegV} is correct. Denote $\Delta'=\Delta+k$. Thus, by Reduction Rule \ref{rule:incidentMandatory}, the size of a maximum independent set in the neighborhood of each vertex in the graph $G$ from which we remove irrelevant edges is bounded by $\Delta'$.

\subsection{The Clique Forest}\label{sec:chordCliqueForest}

Let $F$  denote the clique forest associated with the chordal graphs $G\setminus D'$. Towards bounding the number of leaves in $F$, we need the following lemma.

\begin{lemma}\label{lem:maxIndSet}
Let $I$ be an independent set in the graph $G\setminus D'$. Then, there are most $|D'|\cdot\Delta'$ relevant edges between vertices in $D'$ and vertices in $I$.
\end{lemma}

\begin{proof}
The claim follows from the observation that since $\Delta^I_G\leq\Delta$, $|E_M|\leq k^2$ and $I$ is an independent set, for every vertex $v\in D'$, there are at most $\Delta'$ relevant edges between $v$ and vertices in $I$.
\end{proof}

We will also need the following reduction rule.

\begin{reduction}\label{rule:neiCli}
If there exists a vertex $v$ in $G\setminus D'$ such that the vertices in $N_G(v)$ which are connected to $v$ via relevant edges form a clique, remove the vertex $v$ from $G$.
\end{reduction}

\begin{lemma}
Reduction Rule \ref{rule:neiCli} is safe.
\end{lemma}

\begin{proof}
The special choice of the vertex $v$ implies that every chordless cycle containing $v$ must contain at least one irrelevant edge, and can therefore be ignored. We thus conclude that it is safe to remove the vertex $v$ from $G$.
\end{proof}

The bound on the number of leaves will follow from a bound on the number of bags containing {\em private vertices}, which are defined as follows.

\begin{definition}
A vertex $v$ in $G\setminus D'$ is a {\em private vertex} if there exists only one bag in $F$ that contains it.
\end{definition}

\begin{lemma}\label{lem:boundPrivate}
The number of bags in $F$ containing private vertices is bounded by 
$|D'|\cdot\Delta'$.
\end{lemma}

\begin{proof}
Let $\ell$ denote the number of bags in $F$ containing private vertices. By the definition of a clique forest, if we take exactly one private vertex from each bag in $F$, we obtain an independent set. 
Therefore, the graph $G$ contain an independent set 
of size at least $\ell$. Let $I$ denote this independent set. Observe that since each vertex in $I$ is a private vertex, its neighborhood in $G\setminus D'$ forms a clique. Therefore, after an exhaustive application of Reduction Rule \ref{rule:neiCli}, each vertex in $I$ must be connected by at least one relevant edge to a vertex in $D'$. By Lemma \ref{lem:maxIndSet}, we conclude that $\ell\leq |D'|\cdot\Delta'$.
\end{proof}

We are now ready to bound the number of leaves and nodes of degree at least 3 in the clique forest $F$.

\begin{lemma}\label{lem:leavesDeg3}
Both the number of leaves in $F$ and the number of nodes of degree at least 3 in $F$ are bounded by 
$|D'|\cdot\Delta'$.
\end{lemma}

\begin{proof}
Observe that since every leaf in $F$ corresponds to a maximal clique in $G\setminus D'$, it contains a private vertex. Thus, by Lemma \ref{lem:boundPrivate}, the number of leaves is bounded by 
$|D'|\cdot\Delta'$.
 Since in a forest, the number of nodes of degree at least 3 is bounded by the number of leaves, we conclude that the lemma is correct.
\end{proof}

Next, we turn to bound the size of a bag of $F$, to which end we prove the correctness of the following lemma.


\begin{lemma}\label{lem:boundMaxCli2}
    In polynomial time we can produce an instance $(G',k')$ equivalent to $(G,k)$ such that 
    $k' \leq k$ and the size of any maximal clique in $G'$ is bounded by $\kappa = c\cdot(|\widetilde{D}|^3 \cdot k + |\widetilde{D}| \cdot \Delta' \cdot (k+2)^3)$.
    (Here $c$ is some constant independent of the input.)
\end{lemma}

The proof of this lemma closely follows the approach using which Marx~\cite{Marx10} bounds the size of maximal cliques, but requires a few adaptations since the purpose of the paper~\cite{Marx10} is to give a parameterized algorithm for \cdel, which runs in exponential time, while we need to reduce instances in polynomial time. However, due to the similarity between our proof and this approach, details are deferred to Appendix \ref{app:maxClique}.

Observe that 
the size of each bag of $F$ is bounded by the size of a maximal clique of $G\setminus D'$. Furthermore, since $G\setminus D'$ is a subgraph of $G$, the size of a maximal clique of $G\setminus D'$ is bounded by the size of a maximal clique of $G$.
 Thus, having applied the procedure given by Lemma \ref{lem:boundMaxCli2}, we have the following result.

\begin{lemma}\label{lem:boundMaxCli1}
The size of any bag of $F$ is upper bounded by $\kappa$.
\end{lemma}

\subsection{The Length of Degree-2 Paths}\label{sec:ChordPaths}

{\noindent\bf The Family of Paths $\cal P$.} Let $V_F$ denote the set of each node of degree at least 3 in the forest $F$ as well as each node whose bag has at least one private vertex. Let $\cal P$ denote the set of paths whose endpoints belong to $V_F$ and such that all of their internal nodes do not belong to $V_F$. Clearly, it holds that $|{\cal P}|\leq |V_F|$. By Lemmata \ref{lem:boundPrivate} and \ref{lem:leavesDeg3}, we have the following observation.

\begin{observation}\label{obs:boundP}
$|{\cal P}|\leq 2|D'|\cdot\Delta'$.
\end{observation}

Thus, in light of Lemma \ref{lem:boundMaxCli1}, by bounding the maximum number of nodes on each path in $\cal P$, we can bound the total number of vertices in the graph. To this end, we fix some path $P\in{\cal P}$. Moreover, we orient the path from left to right, where the choice of the leftmost and rightmost nodes is arbitrary.

\bigskip
{\noindent\bf Partitioning the Path $P$.} Next, we will partition $P$ into more ``manageable paths''. To this end, we need the following definition.

\begin{definition}\label{def:comply}
We say that a subpath $Q$ of $P$ complies with a vertex $d\in D'$ if at least one of the following conditions holds.
\begin{enumerate}
\item For every two bags $B$ and $B'$ on $Q$, both $B\subseteq N_G(d)$ and $B'\subseteq N_G(d)$.

\item Let $B_1,B_2,\ldots,B_t$ denote the bags of $Q$ ordered from left to right. Then, at least one of the following condition holds.
	\begin{enumerate}
	\item $B_1\cap N_G(d)\subseteq B_2\cap N_G(d)\subseteq\ldots\subseteq B_t\cap N_G(d)$.
	\item $B_t\cap N_G(d)\subseteq B_{t-1}\cap N_G(d)\subseteq\ldots\subseteq B_1\cap N_G(d)$.
	\end{enumerate}
In particular, $N_G(d)$ is a subset of at least one of the two bags $B_1$ and $B_t$.
\end{enumerate}
\end{definition}

We would like to find a set ${\cal B}$ of at most $\OO(|D'|)$ bags on the path $P$ such that after their removal from $P$, the following lemma will be true.

\begin{lemma}\label{lem:comply}
Each subpath resulting from the removal of the bags in $\cal B$ from $P$ complies with every vertex in $D'$.
\end{lemma} 

The rest of this subsubsection concerns the proof of this lemma. To prove it, it is sufficient to show that for each vertex $d\in D'$, we can find $\OO(1)$ bags such that after their removal from $P$, each of the resulting subpaths complies with $d$. To this end, fix some vertex $d\in D'$.

First, we need the following lemma.

\begin{lemma}\label{lem:inComply1}
Let $u,v\in N_G(d)\setminus D'$ be non-adjacent vertices, $B_u$ be a bag containing $u$ such that no bag to its right (on $P$) contains $u$, and $B_v$ be a bag containing $v$ such that not bag to its left (on $P$) contains $v$. Then, $d$ is adjacent to every vertex in every bag that lies strictly between $B_u$ and $B_v$.
\end{lemma}

\begin{proof}
Let $z$ be a vertex in a bag that lies strictly between $B_u$ and $B_v$. Since this bag is not an endpoint of the path $P$, $z$ belongs to at least two bags that lie between $B_u$ and $B_v$. Denote the path between $B_u$ and $B_v$ by $B_u = B_1 - B_2 - B_3 -  \ldots - B_t = B_v$. Let $B_i$ be the leftmost bag containing $z$, and let $B_j$ be the rightmost bag containing $z$. Suppose, by way of contradiction, that $z\notin N_G(d)$. 

We claim that there exists a path, $P_1$, from $u$ to $z$ such that none of its internal vertices belongs to $B_{i+1} \cup B_{i+2} \ldots \cup B_v$. The proof by induction on the number $i$. When $i=1$, we have that $B_i=B_u$. Then, since the vertices in $B_u$ form a clique, the claim is correct. Now, suppose that the claim holds for $i-1\geq 1$, and let us prove it for $i$.

Since the vertices of the bag $B_i$ form a maximal clique in the graph $G\setminus D'$, and the bag $B_i$ does not contain private vertices, we have that $B_i \subseteq B_{i-1} \cup B_{i+1}$ and $B_i \setminus B_{i+1} \subset B_{i-1}$ is non-empty. Consider a vertex $w \in B_i \setminus B_{i+1}$ and observe that $z$ and $w$ are adjacent in $G$ (since both belong to $B_i$), and further the vertex $w$ is not present in $B_{i+1} \cup B_{i+2} \ldots \cup B_v$. By induction, there is a path $P$ from $u$ to $w$ whose internal vertices do not belong to $B_{i} \cup B_{i+1} \ldots \cup B_v$. Appending the edge $(w,z)$ to $P$ gives us the path $P_1$.

Similarly, there is a path $P_2$ from $z$ to $v$ such that none of its internal vertices belongs to $B_{u} \cup \ldots \cup B_{i-2} \cup B_{i-1}$. Since $i < j$, the paths $P_1$ and $P_2$ have no common vertex except $z$, and there is no edge between an internal vertex of $P_1$ and an internal vertex of $P_2$.

Let $P'_1$ be the subpath of $P_1$ from $u'$ to $z$, where $u'$ is the last vertex in $P_1$ adjacent to $d$. Similarly, let $P'_2$ be the subpath of $P_2$ from $z$ to $v'$, where $v'$ is the first vertex in $P_2$ adjacent to $d$. We may assume that $P'_1$ and $P'_2$ do not contain chords, else we can replace $P'_1$ and $P'_2$ by a chordless subpath of $P_1'$ and  chordless subpath of $P_2'$, respectively, which will still contain $u',z$ and $v'$.
The cycle $d-P'_1-z-P'_2-d$ is a chordless cycle in $G$, contradicting the fact that $G\setminus(D' \setminus \{d\})$ is a chordal graph.
\end{proof}

We also need the following notation. Let $B_\ell$ be the leftmost bag on $P$ that contains a neighbor $v_\ell$ of $d$ such that $v_\ell$ does not belong to any bag to the right of $B_\ell$. Similarly, let $B_r$ be the rightmost bag on $P$ that contains a neighbor $v_r$ of $d$ such that $v_r$ does not belong to any bag to the left of $B_\ell$.

\begin{lemma}\label{lem:inComply2}
Let $B$ and $B'$ be two bags on $P$ that do not lie on the right of $B_\ell$ and such that $B$ lies on the left of $B'$. Then, it holds that $B\cap N_G(d)\subseteq B'\cap N_G(d)$.
\end{lemma}

\begin{proof}
Suppose, by way of contradiction, that $B \cap N_G(d) \not \subseteq B' \cap N_G(d)$. However, this implies that $B$ contains a neighbor $v$  of $d$ such that $v$ does not belong to any bag to the right of $B$, which contradicts the choice of $B_\ell$.
\end{proof}

\begin{lemma}\label{lem:inComply3}
Let $B$ and $B'$ be two bags on $P$ that do not lie on the left of $B_r$ and such that $B$ lies to the right of $B'$. Then, it holds that $B\cap N_G(d)\subseteq B'\cap N_G(d)$. 
\end{lemma}

\begin{proof}
The proof of this lemma is symmetric to the proof of Lemma \ref{lem:inComply2}.
\end{proof}

By Lemmas \ref{lem:inComply1}--\ref{lem:inComply3}, each of the subpaths resulting from the removal of $B_\ell$ and $B_r$ from $P$ complies with $d$. Thus, we conclude that Lemma \ref{lem:comply} is correct.

\bigskip
{\noindent\bf Handling a Manageable Path.} We now examine a subpath of $P$, denoted by $Q$, which complies with every vertex $d\in D'$. We will devise reduction rules such that after applying them exhaustively, the number of vertices in the union of the bags of the path $Q$ will be bounded by $\OO(\kappa)$.

Let $B_1,B_2,\ldots,B_t$ denote the bags of $Q$ ordered from left to right. Moreover, denote $V(Q)=\bigcup_{i=1}^tB_i$ and $A=\bigcap_{i=1}^t B_i$. We partition $D'$ into two sets $D_a$ and $D_p$, where $D_a=\{d\in D': V(Q)\subseteq N_G(d)\}$ and $D_p=D'\setminus D_a$. Here the letters $a$ and $p$ stand for ``all'' and ``partial'', respectively. 
Definition \ref{def:comply} directly implies that the following observation is correct.

\begin{observation}\label{obs:DpWeak}
For every vertex $d\in D_p$, either {\bf (i)} $B_1\cap N_G(d)\subseteq B_2\cap N_G(d)\subseteq\ldots\subseteq B_t\cap N_G(d)$ or {\bf (ii)} $B_t\cap N_G(d)\subseteq B_{t-1}\cap N_G(d)\subseteq\ldots\subseteq B_1\cap N_G(d)$. In particular, either {\bf (i)} $N_G(d)\cap V(Q)\subseteq B_1$ or {\bf (ii)} $N_G(d)\cap V(Q)\subseteq B_t$.
\end{observation}

We also denote $U=V(Q)\setminus(B_1\cup B_t)$. It is sufficient to ensure that $|U|=\OO(\kappa)$ since then, by Lemma \ref{lem:boundMaxCli1}, $|V(Q)|=\OO(\kappa)$. Thus, we can next suppose that $|U|>\delta$ where $\delta=$ \managepath.

For each pair of non-adjacent vertices in $D_a$, we apply Reduction Rule \ref{red:req pair}. The correctness of this operation is given by the following lemma.

\begin{lemma}
Let $u$ and $v$ be two distinct non-adjacent vertices in $D_a$. Then, every solution of size at most $k$ contains at least one of the vertices $u$ and $v$.
\end{lemma}

\begin{proof}
Since $|U|>\delta$ and the size of each bag is bounded by $\kappa$, standard arguments on weighted paths imply that there exists $i\in[t]$ such that $|\bigcup_{j=1}^i B_j|>\delta/2-\kappa$ and $|\bigcup_{j=i+1}^{t}B_j|>\delta/2-\kappa$; indeed, we may iteratively increase $i$ from 1 to $t$ until we reach the first time where it holds that $|\bigcup_{j=1}^i B_j|>\delta/2-\kappa$, in which case it will also hold that $|\bigcup_{j=i+1}^{t}B_j|>\delta/2-\kappa$. Denote $U^1=U\cap((\bigcup_{j=1}^{i-1} B_j)\setminus (B_1\cup B_i))$ and $U^2=U\cap((\bigcup_{j=i+2}^t B_j)\setminus (B_{i+1}\cup B_t))$. Then, again since the size of each bag is bounded by $\kappa$, it holds that $|U^1|,|U^2|>\delta/2-3\kappa\geq k+1$. Moreover, by the definition of a clique forest, $U^1\cap U^2=\emptyset$ and there is no vertex in $U^1$ that is adjacent to a vertex in $U^2$. Thus, for any pair of vertices $x\in U^1$ and $y\in U^2$, the subgraph induced by $\{u,v,x,y\}$ is a chordless cycle. However, any solution of size at most $k$ can only contain at most $k$ vertices from $U^1\cup U^2$, and thus, to hit all of these chordless cycles, it must contain at least one vertex among $u$ and $v$.
\end{proof}

Thus, from now on we can assume that $G[D_a]$ is a clique. However, by the definition of $A$, for every vertex in $A$ and every vertex in $V(Q)$, there exists a bag $B_i$, $i\in[t]$, which contains both of them, and therefore they are adjacent. We thus deduce that the following observation is correct.

\begin{observation}\label{obs:DaAclique}
Any two distinct vertices $v\in D_a\cup A$ and $u\in D_a\cup V(Q)$ are adjacent.
\end{observation}

Let us now examine chordless cycles that contain vertices from $U$.

\begin{lemma}\label{lem:chordlessU1}
Let $C$ be a chordless cycle in $G$ that contains some vertex $u\in U$. Then, no vertex on $V(C)$ belongs to $D_a\cup A$, and both  neighbors of $u$ in $C$ do not belong to $D'\cup A$.
\end{lemma}

\begin{proof}
First, by Observation \ref{obs:DpWeak}, we have that $u$ does not have neighbors (in $G$) in $D_p$, and therefore both neighbors of $u$ in $C$ do not belong to $D_p$. Thus, it remains to show that no vertex in $V(C)$ belongs to $D_a\cup A$. By Observation \ref{obs:DaAclique}, if at least one of vertex $v\in V(C)$ belongs to $D_a\cup A$, it is adjacent to $u$ in $G$ and it is either a neighbor of $u$ in $C$ adjacent in $G$ to the other neighbor of $u$ in $C$ or it is adjacent in $G$ to both neighbors of $u$ in $C$. In any case, if at least one of vertex $v\in V(C)$ belonged to $D_a\cup A$, the cycle $C$ would have contained a chord, contradicting the supposition that $C$ is a chordless cycle. We thus conclude that the lemma is correct.
\end{proof}

\begin{lemma}\label{lem:chordlessU2}
Let $C$ be a chordless cycle in $G$ that contains some vertex $u\in U$. Then, $C$ contains a path between a vertex in $B_1\setminus A$ and a vertex in $B_t\setminus A$ whose internal vertices belong to $U$ and one of them is $u$.
\end{lemma}

\begin{proof}
Since $D'$ is an approximate solution, the cycle $C$ must contain at least one vertex that does not belong to $U$. Thus, by Lemma \ref{lem:chordlessU1}, we deduce that the cycle $C$ contains two subpaths, each between $u$ and a vertex in $(B_1\cup B_t)\setminus A$, whose only common vertex is $u$ and whose internal vertices belong to $U$. Moreover, one of these paths must contain an endpoint from $B_1\setminus A$ and the other from $B_t\setminus A$, else the cycle $C$ contains a chord corresponding to the edge between these two endpoints. 
Therefore, by concatenating the two paths, we obtain the desired subpath of $C$.
\end{proof}

We continue our examination of chordless cycles that contain vertices from $U$ in the context of separators.

\begin{lemma}\label{lem:cut}
Let $S$ be a minimal solution that contains at least one vertex from $U$. Then, there exists $i\in[t-1]$ such that {\bf (i)} $(B_i\cap B_{i+1})\setminus A\subseteq S$, and {\bf (ii)} $S\cap U\subseteq B_i\cap B_{i+1}$.
\end{lemma}

\begin{proof}
{\noindent\bf Property (i).} Since $S$ is a minimal solution, $G$ contains a chordless cycle $C$ with a vertex $u\in S\cap U$ and no other vertex from $S$. By Lemma \ref{lem:chordlessU2}, $C$ contains a path between a vertex in $B_1\setminus A$ and a vertex in $B_t\setminus A$ whose set of internal vertices includes $u$ and is a subset of $U$. In particular, since each bag induces a clique while $C$ is a chordless cycle, $C$ contains a path $P$ between a vertex $x\in B_1\setminus A$ and a vertex in $y\in B_t\setminus A$ whose set of internal vertices is a non-empty subset of $V(Q)\setminus A$, and such that $(V(C)\setminus V(P))\cap V(Q)=\emptyset$. Thus, $x$ and $y$ are not adjacent. Moreover, by Lemma \ref{lem:chordlessU1}, $V(C)\cap (D_a\cup A)=\emptyset$. Thus, since $D'$ is a redundant approximate solution, $C$ contains a vertex $d\in D_p$. 


Suppose, by way of contradiction, that there does not exist $i\in[t-1]$ such that $(B_i\cap B_{i+1})\setminus A\subseteq S$. Then, $G$ has a path $P'$ between $x$ and $y$ whose internal vertices belong to $V(Q)\setminus(S\cup A)$. Since $(V(C)\setminus V(P))\cap V(Q)=\emptyset$, it holds that $(V(C)\setminus V(P))\cap V(P')=\emptyset$, which implies that $C$ contains a path between $x$ and $y$ whose set of internal vertices is disjoint from the one of $V(P')$. Observe that if a graph $H$ contains a vertex $a$ with two non-adjacent neighbors, $b$ and $c$, such that $H\setminus\{a\}$ has a chordless path between $b$ and $c$ with at least one vertex that is not a neighbor of $a$, then $H$ has a chordless cycle. On the one hand, by the definition of a clique forest, any vertex in $V(G)\setminus (V(Q)\cup D')$ cannot be adjacent to both a vertex in $B_1\setminus A$ and a vertex in $B_t\setminus A$, and it is adjacent to no vertex in $U$. On the other hand, Observation \ref{obs:DpWeak}  implies that any vertex in $D_p$ also satisfies this property. Thus, $V(P')\cap U=\emptyset$, since any vertex in this set fits the above description of the vertex $a$ where $H=G[(V(C)\setminus V(P))\cup V(P')]$, which is a subgraph of the chordal graph $G\setminus S$. Without loss of generality, we have that $P'$ contains a subpath $p-q-r$ where $p,q\in B_1\setminus A$ and $r\in B_t\setminus A$. Let $P''$ denote a shortest path between $p$ and $r$ in $G[(V(C)\setminus V(P))\cup (V(P')\setminus\{q\})]$, which contains at least two internal vertices (since it must contain a vertex from $V(C)\setminus V(P)$, else $P'$ would have had a chord, and such a vertex cannot be adjacent to both $p$ and $r$). Every vertex on $P''$ should be adjacent to $q$, else $q$ fits the above description of the vertex $a$. Since $P'$ is a chordless path and $(V(C)\setminus V(P))\cap V(Q)=\emptyset$, $P''$ does not contain any vertex from $B_1$ that is not $p$, and the only vertex from $B_t$ that is not $r$ and which $P''$ can contain is the neighbor of $r$. Overall, this implies that $P''$ contains a vertex from $V(C)\setminus V(P)$ adjacent to both $q$ and a vertex in $B_t\setminus A$, which is a contradiction.

\medskip
{\noindent\bf Property (ii).} It remains to show that $S\cap U\subseteq B_i\cap B_{i+1}$. To this end, we consider some arbitrary vertex $u\in S\cap U$ and show that it belongs to $B_i\cap B_{i+1}$. Since $S$ is a minimal solution, $G$ contains a chordless cycle $C'$ such that $V(C')\cap S=\{u\}$. By Lemma \ref{lem:chordlessU2}, $C'$ contains a subpath between a vertex in $B_1\setminus A$ and a vertex in $B_t\setminus A$ whose internal vertices belong to $U$. Thus, by the definition of a clique forest, $C'$ must contain a vertex from $(B_i\cap B_{i+1})\setminus A$. However, we have already shown that $(B_i\cap B_{i+1})\setminus A\subseteq S$, and therefore it must hold that $u\in B_i\cap B_{i+1}$ (since otherwise we reach a contradiction to the fact that $V(C')\cap S=\{u\}$).
\end{proof}

Let ${\cal W}$ be the family of each subset $W\subseteq (B_1\cup B_t)\setminus A$ of size at most $k$ for which there exists an index $i\in[t-1]$ such that $W=(B_i\cap B_{i+1})\setminus(A\cup U)$ and $|(B_i\cap B_{i+1})\cap U|\leq k$. We can easily bound the size of the family ${\cal W}$ as follows. 

\begin{lemma}\label{lem:boundW}
$|{\cal W}|\leq 2k+1$. 
\end{lemma}

\begin{proof}
Let $i$ be the smallest index in $[t-1]$ for which there exists $W^i\in {\cal W}$ such that $W^i=(B_i\cap B_{i+1})\setminus(A\cup U)$, and let $j$ be the largest or which there exists $W^j\in {\cal W}$ such that $W^j=(B_i\cap B_{i+1})\setminus(A\cup U)$. Then, by the definition of a clique forest, for every set $W\in {\cal W}$ it holds that $W\subseteq W^i\cup W^j$. Furthermore, the sets in ${\cal W}$ can be sorted by $W_1,W_2,\ldots,W_{|{\cal W}|}$ such that for all $r\in[|{\cal W}|-1]$, $W_{r+1}\cap B_1\subseteq W_r\cap B_1$ and $W_r\cap B_t\subset W_{r+1}\cap B_t$. We thus conclude that $|{\cal W}|\leq 2k+1$. 
\end{proof}

We proceed by associating a separator with each set $W\in {\cal W}$ as follows. First, let ${\cal I}_W$ denote the set of all indices $i\in[t-1]$ such that $W=(B_i\cap B_{i+1})\setminus(A\cup U)$. Now, let $i_W$ denote an index in ${\cal I}_W$ that minimizes $|(B_i\cap B_{i+1})\cap U|$ (if there are several choices, choose one arbitrarily). We further denote $M=\bigcup_{W\in{\cal W}}\left((B_{i_W}\cap B_{i_W+1})\cap U\right)$. Observe that by Lemmata \ref{lem:boundMaxCli1} and \ref{lem:boundW}, $|M|=\OO(k^2)$. Thus, it is sufficient to argue that there exists a vertex in $U\setminus M$ that can be removed from $G$ (since as long as $|U|>\delta$, we will be able to find such a vertex). To this end, we will need the following lemma.

\begin{lemma}\label{lem:irrelUM}
Let $u\in U\setminus M$. If $(G,k)$ is a yes-instance, then it has a solution $S$ of size at most $k$ that does not contain the vertex $u$.
\end{lemma}

\begin{proof}
Suppose that $(G,k)$ is a yes-instance, and let $S$ be a solution of minimum size. Assume that $u\in S$, else we are done. By Lemma \ref{lem:cut}, there exists $i\in[t-1]$ such that $(B_i\cap B_{i+1})\setminus A\subseteq S$ and $S\cap U\subseteq B_i\cap B_{i+1}$.
Denote $W=(B_i\cap B_{i+1})\setminus(A\cup U)$. Since $|S|\leq k$ and $W\subseteq S$, we have that $W\in{\cal W}$. Denote $R=(B_{i_W}\cap B_{i_W+1})\cap U$ and $T=(B_i\cap B_{i+1})\cap U$. Since $u\in U\setminus M$, it holds that $u\notin R$. Moreover, since $S\cap U\subseteq B_i\cap B_{i+1}$, it holds that $u\in T$. By the definition of $i_W$, we have that $|R|\leq |T|$. Thus, to show that the lemma is correct, it is sufficient to show that $S'=(S\setminus T)\cup R$ is a solution.

Suppose, by way of contradiction, that $S'$ is not a solution. Then, since $S$ is a solution of minimum size, there exist a chordless cycle $C$ and a vertex $v\in T$ such that $v\in V(C)$ and $V(C)\cap S'=\emptyset$. Since $T\subseteq U$, by Lemma \ref{lem:chordlessU2}, $C$ contains a path between a vertex in $(B_1\cap B_2)\setminus A$ and a vertex in $(B_{t-1}\cap B_t)\setminus A$ whose internal vertices belong to $U$. In particular, by the definition of a clique forest, $V(C)\cap ((B_{i_W}\cap B_{i_W+1})\setminus A)\neq\emptyset$. However, we have that $(B_{i_W}\cap B_{i_W+1})\setminus A = R\cup W\subseteq S'$, which contradicts the fact that $V(C)\cap S'=\emptyset$.
\end{proof}

We are now ready to present our reduction rule.

\begin{reduction}\label{rule:UM}
Let $u\in U\setminus M$. Remove the vertex $u$ from the graph $G$ and add an edge between any two non-adjacent vertices in $N_G(u)$.
\end{reduction}

\begin{lemma}
Reduction Rule \ref{rule:UM} is safe.
\end{lemma}

\begin{proof}
Let $G'$ be the graph resulting from the application of this rule. For the forward direction, suppose that $(G,k)$ is a yes-instance, and let $S$ be a solution of minimum size. By Lemma \ref{lem:irrelUM}, we can assume that $u\notin S$, and therefore $S\subseteq V(G')$. Thus, to show that $(G',k)$ is a yes-instance, we need to prove that $S$ hits every chordless cycle in $G'$. Let $C$ be a chordless cycle in $G'$. Suppose that this cycle does not exist in $G$, else it is clear that $S$ hits it. Then, $C$ contains an edge between two vertices $v,w\in N_G(u)$ that are non-adjacent in $G$. Observe that since $C$ is a chordless cycle and $G'[N_G(u)]$ is a clique, $V(C)\cap N_G(u)=\{v,w\}$. Thus, by replacing $\{v,w\}$ by $\{v,u\}$ and $\{u,w\}$, we obtain a chordless cycle in $G$. Since $S$ hits this cycle and $u\notin S$, it holds that $S$ also hits $C$.

For the backward direction, suppose that $(G',k)$ is a yes-instance, and let $S$ be a solution of minimum size. To show that $(G,k)$ is a yes-instance, it is sufficient to show that $S$ is also a solution to $(G,k)$. Since $V(G')\subseteq V(G)$, it is clear that $S\subseteq V(G)$. We need to prove that $S$ hits every chordless cycle in $G$. Let $C$ be a chordless cycle in $G$. Suppose that this cycle does not exist in $G'$, else it is clear that $S$ hits it. 
Furthermore, suppose that $G'[V(C)\setminus\{u\}]$ does not contain a chordless cycle, else again it is clear that $S$ hits $C$. We get that $C$ contains two vertices $v,w\in N_G(u)$ that are not adjacent in $G$. Since $u\in U$, Observation \ref{obs:DpWeak} and the definition of a clique forest imply that $N_G(u)\subseteq V(Q)\cup D_a$. By Observation \ref{obs:DaAclique}, we deduce that $v,w\notin D_a\cup A$, and that $C$ contains at most one vertex from $D_a\cup A$ (since any two vertices from $D_a\cup A$ are adjacent to each other and to both $v$ and $w$). Since $D'$ is a redundant approximate solution, $C$ must contain a vertex $p\in D_p$. Since $G'[V(C)\setminus\{u\}]$ does not contain a chordless cycle, the neighbors of $p$ on $C$ belong to $N_G(u)$, and we can assume w.l.o.g that these neighbors are $v$ and $w$. However, since $v$ and $w$ are not adjacent there cannot be a bag that contains both of them, which results in a contradiction to Observation \ref{obs:DpWeak}.
\end{proof}

\subsection{Unmarking Irrelevant and Mandatory Edges}\label{sec:unmark}
Recall that our instance includes a set $E_I$ of irrelevant edges and a set $E_M$ of mandatory edges. It is clear that we can unmark each irrelevant edge (these edges were marked only for the sake of clarity of the analysis of our kernel). However, to unmark mandatory edge, we need the following operation.

\begin{reduction}\label{rule:unmarkMandatory}
For every mandatory edge $\{x,y\}$ introduce $k+1$ pairs of new vertices, $\{x_1, y_1\}, \{x_2, y_2\},\ldots,\{x_{k+1},y_{k+1}\}$, and for each pair $\{x_i,y_i\}$ add the edges $\{x,x_i\}, \{x_i,y_i\}$ and $\{y_i,y\}$. Moreover, unmark the edge $\{x,y\}$.
\end{reduction}

\begin{lemma}
Reduction Rule \ref{rule:unmarkMandatory} is safe.
\end{lemma}

\begin{proof}
Let $(G',k,E_I=\emptyset,E_M')$ be the instance resulting from the application of this rule. Each edge $\{x,y\}$ is contained (in $G'$) in $k+1$ cycles of size $4$ which do not share vertices other than $x$ and $y$. Therefore, and solution of size at most $k$ to $(G',k)$ must contain at least one of the vertices $x$ and $y$. Thus, since any chordless cycle in $G$ is also present in $G'$ and we have only unmarked the edge $\{x,y\}$, we conclude that if $(G',k,E_I,E_M')$ is a yes-instance, so is $(G,k,E_I,E_M)$. On the other hand, any solution to $(G,k,E_I,E_M)$ hits all of the chordless cycles in $G'$ since each of these chordless cycles is either present in $G$ or contain both of the vertices $x$ and $y$. Therefore, if $(G,k,E_I,E_M)$ is a yes-instance, so is $(G',k,E_I,E_M')$.
\end{proof}

Recall that $|E_M|\leq k^2$. Hence, the total number of newly added vertices does not exceed $\OO(k^3)$.

\subsection{The Number of Vertices in the Kernel}\label{sec:chordProof}

In this section, we obtained an approximate solution $\widetilde{D}$ of size $f(k)$ and a redundant approximate solution $D$ of size $\OO(k\cdot f(k))$. Then, we examined the clique forest $F$ associated with the chordal graph $G\setminus D'$ where $|D'|=\OO(|D|+k^2)$. To this end, we considered a set $\cal P$ of degree-2 paths that together cover all of the nodes of the forest, and showed that $|{\cal P}|=\OO(|D'|\cdot\Delta')$. Recall that $\Delta'=\OO(k\cdot f(k))$. We removed $\OO(|D'|)$ bags, each of size $\kappa=\OO(|\widetilde{D}|^3 \cdot k + |\widetilde{D}| \cdot \Delta' \cdot k^3)$, from each path $P\in{\cal P}$, and considered each of the resulting subpaths $Q$. We showed the number of vertices in the union of the bags of the path $Q$ will be bounded by $\OO(\kappa)$. Finally, we added $\OO(k^3)$ new vertices to unmark mandatory edges. Thus, we conclude that the number of vertices in our kernel is bounded by 

\smallskip
$\begin{array}{ll}
&\OO(|{\cal P}|\cdot |D'|\cdot \kappa)\\
= & \OO(|D'|^2\cdot\Delta'\cdot(|\widetilde{D}|^3 \cdot k + |\widetilde{D}| \cdot \Delta' \cdot k^3))\\
=& \OO(f(k)^3k^3\cdot(f(k)^3k + f(k)^2k^4))\\
=& \OO(f(k)^5k^4\cdot(f(k) + k^3))
\end{array}$

\smallskip
Recall that by Lemma \ref{cor:newApprox}, we can assume that $f(k)=$ \NEWinitapproxsolsizek. Thus, at this point, we obtain a kernel of size \NEWinitsizeker.

\subsection{A Better Kernelization Algorithm}\label{sec:better}

Finally, we present a bootstrapping trick that will exploit the nature of our approximation algorithm to obtain a kernel of size \NEWsizeker. Recall that at this point, where we have already run our kernelization algorithm once, it holds that $n=$ \NEWinitsizeker. Now, we again recall that \cdel admits an \NEWapproxfactor-factor approximation algorithm \cite{manuscript}. Currently, it holds that $f(k)=$ \NEWapproxsolsizek\ rather than $f(k)=$ \NEWinitapproxsolsizek. Thus, if we rerun our kernelization procedure, obtaining a kernel of size $\OO(f(k)^5k^4\cdot(f(k) + k^3))$ (see Section \ref{sec:chordProof}), it now holds that this size is upper bounded by \NEWsizeker. This concludes the proof of correctness of Theorem \ref{thm:chordal}.

\section{Conclusion}
In this paper we obtained a polynomial kernel for \cdel\ of size \NEWsizeker. The new kernel significantly improves over the previously known $\OO(k^{161}\log^{58}k)$ sized kernel. 
We believe that the notion of independence degree and the bootstrapping trick used in our kernelization procedure could be useful in designing polynomial kernels for other {\sc $\cal F$-Vertex (Edge) Deletion} problems, where  $\cal F$ is characterized by 
an infinite set of forbidden induced graphs. We conclude the paper with the following open problems. 
\begin{itemize}
\item Design a polynomial kernel for \cdel\ of size $\OO(k^c)$ for some fixed constant $c\leq 5$.
\item Design a constant-factor approximation algorithm for \cdel.
\item Does there exist an FPT algorithm for \cdel\ with running time $c^k n^{\OO(1)}$, for some fixed constant $c$?
\end{itemize}

\bibliographystyle{siam}
\bibliography{references}


\appendix


\section{Bounding the Size of a Maximal Clique}\label{app:maxClique}
 

Let $(G,k)$ be an instance of \cdel whose independent degree is bounded $\Delta$,
    and let $D$ be some (possibly approximate) solution to this instance.
In this appendix, we will bound the size of a maximal clique in the graph $G$.
We will show that if the graph admits a chordal deletion set of size at most $k$, then for any maximal clique $K$ in $G \setminus  D$, it is safe to remove all but a bounded number vertices of $K$.
Formally, we will show the following lemma.
\begin{lemma}\label{lem:boundMaxCli}
    Let $(G,k)$ be an instance of \cdel whose independent degree is bounded $\Delta$, and let $D$ be a solution to this instance.
    Then in polynomial time we can produce an equivalent instance $(G',k')$ such that 
    $k' \leq k$, and the size of any maximal clique in $G'$ is bounded by $c\cdot(|D|^3 \cdot k + |D| \cdot \Delta \cdot (k+2)^3)$.
    (Here $c$ is some constant independent of the input.)
\end{lemma}

Our proof of this lemma is an adaptation of the work of Marx~\cite{Marx10} with a few modifications.
Specifically, the lemmas in~\cite{Marx10} construct a so called ``necessary set'' of vertices with the property that one of the vertices in this set must be part of any solution of size $k$.
We modify these lemmas to ensure that a necessary set output by them always has at most two vertices. 
Observe that such a necessary set is either a vertex that must be part of any solution of size at most $k$, or a mandatory edge.
We note that Jansen and Pilipczuk \cite{JansenPili2016} also give similar result, inspired by the results of ~\cite{Marx10}.
Given a redundant approximate solution $\widehat{D}$, the size of a maximal clique in $G \setminus \widehat{D}$ can be bounded by $\OO(|\widehat{D}|^3k)$.
However the present method for computing a redundant approximate solution implies that $|\widehat{D}| \geq |D| \cdot k$, and hence Lemma~\ref{lem:boundMaxCli} gives a better upper-bound.

\begin{lemma}[\cite{JansenPili2016}]\label{lem:boundMaxCli-old}
    The size of each bag in the clique forest $F$ is bounded by $c\cdot|D'|^3k$, where $c$ is some constant independent of the input.
\end{lemma}

For the rest of this subsection, we fix a maximal clique $K$ of $G \setminus D$, which contains more than $c\cdot(|D|^3 \cdot k + |D| \cdot \Delta \cdot (k+2)^3)$ vertices.
We will show that we can mark a bounded number of vertices of $K$ so that the following holds.
Let $X$ be any set of at most $k$ vertices such that $G \setminus X$ has a chordless cycle $H$ that contains a vertex $u$ of $K$.
If $u$ is an unmarked vertex then there is another chordless cycle $H'$ in $G \setminus X$ that avoids $u$ and contains strictly fewer unmarked vertex in $K$.
This condition implies that we can safely ignore any chordless cycle that includes an unmarked vertex, which further implies that it is safe to delete these vertices from the graph.
We shall closely follow the notations and proofs of ~\cite{Marx10}, but in light of the bound on (relevant) independent-degree of vertices in $D$, we shall modify them appropriately.
For a vertex $v \in V(G) \setminus D$, we say that the vertex has the \emph{label} $t \in D$
if $v$ is a neighbor of $t$.
Note that the edge $(t,v)$ could be relevant or irrelevant, and a vertex may have several labels depending on its set of neighbors in $D$.

\subsection{Dangerous vertices and their witnesses in $K$}
Let us begin with the notion of dangerous vertices for the clique $K$ in the graph.
\begin{definition}
    Let $v \in V(G) \setminus (D \cup K)$ be a relevant neighbor of $t \in D$ such that there is a path $P$ from $v$ to $u \in K$ whose internal vertices do not have the label $t$.
    \begin{enumerate}[(i)]
        \item The vertex $v$ is a \emph{$t$-dangerous vertex} for $K$ if the vertex $u$ does not have the label $t$.
        
        \item The vertex $v$ is a \emph{$t^*$-dangerous vertex} for $K$ if the vertex $u$ is also a relevant neighbor of $t$ and $u,v$ are not neighbors in $G$.
    \end{enumerate}
    The vertex $u \in K$ is called a \emph{$t$-witness ($t^*$-witness)} of $v$,
    and the path $P$ is called a \emph{$t$-witness ($t^*$-witness) path} of $v$.
\end{definition}
In \cite{Marx10}, it is shown that there are many vertices in the clique $K$ whose deletion does not affect any of the dangerous vertices in the following sense.
We mark a bounded number of vertices in $K$ such that, for any subset $X$ of vertices of size at most $k$, if there is chordless cycle in $G \setminus  X$ that passes through $t$, a dangerous vertex $v$ and through some unmarked vertex $u$ in $K$, then there is another chordless cycle in $G \setminus X$ which contains a marked vertex $u'$ and avoids $u$.
This implies that we may ignore any chordless cycles in $G$ that includes an unmarked vertex of $K$.
Note that the definition above differs from \cite{Marx10} in the requirement that $v$ (and $u$) must be a relevant neighbor of $t$, since if the edge $(t,v)$ (or the edge $(t,u)$) were irrelevant then any chordless cycles that contain both $t$ and $v$ (or $t$ and $u$) can be safely ignored.
We have the following bounds on the size of an independent set of dangerous vertices in $G$.

\begin{lemma}[Lemma 11,~\cite{Marx10}]
    If $I$ is any collection of independent $t$-dangerous vertices, then either $|I| \leq 6k^2$
    or we can find a new mandatory edge in the graph, or we can find a vertex that must be part of any solution of size at most $k$.
\end{lemma}

The following lemma improves upon Lemma 12 of \cite{Marx10} by using the bound on the independent degree of vertices in $D$.
\begin{lemma}[Lemma 12,~\cite{Marx10}]
    If $I$ is any collection of independent $t^*$-dangerous vertices, then $|I| \leq \Delta$.
\end{lemma}

The following lemma shows that if $Q$ is a clique of $t$-dangerous vertices we require only $k+1$ vertices as witnesses for all the vertices of $Q$.
\begin{lemma}[Lemma 13 \cite{Marx10}]
    Let $Q$ be a clique of $t$-dangerous vertices.
    Then we can mark $k+1$ vertices in $K$ such that for any set $X$ of $k$ vertices,
    if $v \in Q$ has an unmarked $t$-witness in $K \setminus X$ then it has a marked $t$-witness in $K \setminus X$.
\end{lemma}

In the context of the following lemma, we remark that there is a similar lemma for cliques of $t^*$-dangerous vertices in \cite{Marx10}, but we give a version of it that is more suitable for our purposes.
\begin{lemma}[Lemma 14,~\cite{Marx10}]
    Let $Q$ be a clique of $t^*$ dangerous vertices. 
    Then either we can find a vertex that must be part of any solution of size at most $k$, or we can find a new mandatory edge, or else we can mark $(k+2)^3$ vertices of $K$ such that for any set $X$ of $k$ vertices, if $v \in Q$ has a unmarked $t^*$ witness in $K \setminus X$ then it also has a marked witness $K \setminus X$.
\end{lemma} 
\begin{proof}
    Our proof is a minor modification of the proof of Lemma 14,~\cite{Marx10}.
    Let $T$ be a clique-tree of $G \setminus D$.
    Following \cite{Marx10}, we say that a vertex $v$ covers a bag $x$ of $T$ if $v$ is contained in the bag $x$,
    and then $T_v$ denotes the sub-clique-tree of all the bags which are covered by $v$.
    Now, since $Q$ and $K$ are cliques in $G \setminus D$, there are two bags $x$ and $y$ which contain $Q$ and $K$, respectively.
    Consider the unique path connecting $x$ and $y$ in $T$, and suppose that the bags on this path are numbered as  $x = 1, 2, \ldots s = y$.
    Let $u_1, u_2, \ldots$ be vertices of $K$ with label $t$, and let $a_i$ denote the smallest numbered bag on this path that occurs in $T_{u_i}$, i.e. the smallest numbered bag containing $u_i$.
    Similarly, let $v_1, v_2, \ldots$ be the vertices of $Q$ and let $b_i$ denote the largest numbered bag which is contained in $T_{v_i}$. 
    It follows that $T_{v_i}$ and $T_{u_j}$ intersect if and only if $a_j \leq b_i$,
    in which case there is an edge $(v_i,u_j)$ in the graph.
    We also assume that the collection of $a_i$ and  $b_i$ are distinct, which is easily achieved by adding additional bags along this path in $T$.
    Further we assume that the vertices $u_1,u_2, \ldots$ and $v_1, v_2, \ldots$ are ordered so that the sequences of $a_i$ and $b_j$ are strictly increasing.
    
    We define a subsequence of $b_i$ and $a_j$ as follows.
    Let $\beta_1 = 1$, and for ever $j \geq 1$ let $\alpha_j$ be the smallest value such that $a_{\alpha_j} > b_{\beta_j}$. 
    For every $i \geq 2$, let $\beta_i$ be the smallest value such that $b_{\beta_i} > a_{\alpha_{i-1}}$.
    Observe that we obtain a strictly increasing sequence, $b_{\beta_1} < a_{\alpha_1} < b_{\beta_2} < a_{\alpha_2} \ldots $.
    Let $\beta_\ell$ be the last element of the above sequence which corresponds to a vertex in $Q$.
    
    Now, let $u_s$ be a witness of a $t^*$-dangerous vertex $v_{\beta_j}$.
    We will show that $a_{\alpha_j}$ is also a witness for $v_{\beta_j}$.
    Clearly, $a_s > b_{\beta_j}$, which implies $a_s \geq a_{\alpha_j}$.
    Hence, the $t^*$-witness path from $v_{\beta_j}$ to $u_s$ passes through the bag $a_{\alpha_j}$ which contains $u_{\alpha_j}$ (see the proof of Lemma 14,~\cite{Marx10}) .
    Further, $a_{\alpha_j} > b_{\beta_j}$ implies that $v_{\beta_j}$ and $a_{\alpha_j}$ are not neighbors in $G$.
    Therefore, $a_{\alpha_j}$ is also a witness for $v_{\beta_j}$.
    
    Now, suppose that $\ell \leq (k+2)^2$. 
    Then for each $i = 1, 2, \ldots, \ell$ we mark the $k+2$ vertices $u_{\alpha_i}, u_{\alpha_i + 1}, \ldots, u_{\alpha_i + k + 1}$ (if they exist), and we will show that this set of marked vertices contains a sufficient number of witnesses for every vertex in $Q$.
    Note that we have marked at most $(k+2)^3$ vertices of $K$.
    Consider any set $X$ of $k$ vertices such that a vertex $v_x \in Q$ has a witness path in $G \setminus X$ to some unmarked $u_y \in K$.
    In other words there is a chordless cycle in $G \setminus X$ which includes the vertex $t$, a $t^*$-dangerous vertex $v_x$ and its $t^*$-witness $u_y \in K$.
    Since $v_x$ and $u_y$ are non-neighbors, we have that $b_x < a_y$, which implies that there is some $j$ for which $b_x < a_{\alpha_j} \leq a_y$.
    If $a_y$ is not marked then $y > \alpha_j + k + 1$, and hence $G \setminus X$ contains some $b_x < a_{\alpha_j + r} < a_y$.
    Now the witness path from $v_x$ to $a_y$ must contain a neighbor of $a_{\alpha_j + r}$,
    which implies that $a_{\alpha_j + r}$ is also a witness for $v_x$, in $G \setminus X$.
    
    Finally, suppose that $\ell > (k+2)^2$.
    We will show that either the vertex $t$ must be part of any solution of size $k$, or we can find a new mandatory edge in the graph.
    Let $P_i$ be a witness path from $v_{\beta_i}$ to $u_{\alpha_i}$ for every $1 \leq i \leq \ell$.
    Consider the collection of chordless cycles $H_{(k+2) \cdot j}$ for $1 \leq j \leq k+1$,
    where $H_{(k+2) \cdot j} = (t \, v_{\beta_{(k+2) \cdot j}} \, P_{(k+2) \cdot j} \, u_{\alpha_{(k+2) \cdot j}} \, t)$.
    If all these cycles are pairwise vertex disjoint except for the vertex $t$, then we have obtained a collection of $k+1$ chordless cycles with $t$ as the only common vertex.
    Hence $t$ must be part of any solution of size at most $k$.
    
    Otherwise, let $H_{(k+2) \cdot j}$ and $H_{(k+2) \cdot j'}$ have a common vertex $w$, for some $j < j'$.
    Then observe that $w$ can only be an internal vertex of the paths $P_{(k+2) \cdot j}$ and $P_{(k+2) \cdot j'}$, and therefore it is not a neighbor of $t$.
    Since $T$ is a tree-decomposition, this implies that $w$ occurs in every bag numbered between $a_{\alpha_{(k+2) \cdot j}}$ and $b_{\beta_{(k+2) \cdot j'}}$.
    Since $ a_{\alpha_{(k+2) \cdot j}} < b_{\beta_{(k+2) \cdot j + 1}} < a_{\alpha_{(k+2) \cdot j + 1}} < \ldots\ldots < b_{\beta_{(k+2) \cdot j + k + 1}} < a_{\alpha_{(k+2) \cdot j + k + 1}} < b_{\beta_{(k+2) \cdot j'}}$, we have that $w$ occurs in every one of these bags.
    Observe that we obtain a collection of $k+1$ chordless cycles of length $4$, namely
    $(t \, v_{\beta_{(k+2)\cdot j + i}}\, w \, u_{\alpha_{(k+2)\cdot j + i}} \,t)$ for $i = 1,2, \ldots k+1$, such that the vertices $t$ and $w$ are the only common vertices.
    Since $t$ and $w$ are non-adjacent in $G$, we obtain a new mandatory edge $(t,w)$ in $G$.
\end{proof}

Since $G \setminus D$ is a chordal graph which contains all the dangerous vertices for $K$, it follows that the graph induced by the set of all $t$-dangerous ($t^*$-dangerous) vertices forms a chordal graph as well.
Since a chordal graph is perfect, it has a clique cover of size $\alpha$, where $\alpha$ is the cardinality of a maximum independent set in the graph (see~\cite{Golumbic80}).
This gives us the following lemma, using the bound on the size of an independent set of dangerous vertices.
\begin{lemma}[Lemmas 15 \& 16, ~\cite{Marx10}]
    \begin{enumerate}[(i)]
        \item Either we can find a vertex that must be part of any solution of size at most $k$, or we can find a new mandatory edge, or we can mark $6k^2(k+1)$ vertices in $K$ such that for any set $X$ of $k$ vertices, if a $t$-dangerous vertex $v$ has a unmarked witness in $K \setminus X$, then it also has a marked witness in $K \setminus X$.
        
        \item Either we can find a vertex that must be part of any solution of size at most $k$, or we can find a new mandatory edge, or we can mark $ \Delta \cdot (k+2)^3$ vertices in $K$ such that for any set $X$ of $k$ vertices, if a  $t^*$-dangerous vertex has a unmarked witness in $K \setminus X$, then it also has a marked witness in $K \setminus X$.
    \end{enumerate} 
\end{lemma}

From the above lemma we conclude that, for any set $X$ of at most $k$ vertices, to test if $X$ intersects all those chordless cycles that pass through $t$, $K$ and some dangerous vertex $v$, it is sufficient to consider only the marked vertices of $K$ as witnesses.
However the above lemma considers only a single vertex $t \in D$, and we must mark additional witness vertices for each $t \in D$.

\begin{lemma}[Lemmas 15 \& 16, ~\cite{Marx10}]\label{lem:cliMark1}
    \begin{enumerate}[(i)]
        \item Either we can find a vertex that must be part of any solution of size at most $k$, or we can find a new mandatory edge, or we can mark $|D| \cdot 6k^2(k+1)$ vertices in $K$ such that for any set $X$ of $k$ vertices, if a $t$-dangerous vertex $v$ has a unmarked witness in $K \setminus X$, then it also has a marked witness in $K \setminus X$.
        
        \item Either we can find a vertex that must be part of any solution of size at most $k$, or we can find a new mandatory edge, or we can mark $|D| \cdot \Delta \cdot (k+2)^3$ vertices in $K$ such that for any set $X$ of $k$ vertices, if a $t^*$-dangerous vertex has a unmarked witness in $K \setminus X$, then it also has a marked witness in $K \setminus X$.
    \end{enumerate}
\end{lemma}

\subsection{Fragments of Chordless cycles intersecting $K$}
Now we shall mark vertices for chordless cycles in $G$ that intersect $K$.
If $H$ is a chordless cycle in $G$, then consider $F, P_1, P_2, \ldots, P_s$ where $F= H \cap D$ and $P_1, P_2 \ldots, P_s$ are the paths in $H \setminus D$.
It follows that each $P_i$ has exactly two labels from $F$ on its endpoints, and the internal vertices have no labels from $F$, and further these paths are pairwise independent (i.e. there is no edge between two vertices that are in two different paths).
We call $F, P_1, \ldots ,P_s$ the \emph{fragments} of the chordless cycle $H$.
Since $K$ is a clique, at most one of these paths intersects $K$, which we assume to be the path $P_1$, and further it contains at most two vertices from $K$.
We will show that we can mark a bounded number of vertices in $K$ such that for any chordless cycle $H$ that includes an unmarked vertex (that lies in $P_1$), there is another chordless cycle $H'$ that avoids this unmarked vertex.
Let us first consider the case when $P_1$ just a single vertex (in $K$).
The following lemma is a close variant of Lemma 18,~\cite{Marx10} for this case.

\begin{lemma}[Lemma 18,~\cite{Marx10}]\label{lem:cliMark2}
    Let $F, P_1, \ldots, P_s$ be the fragments of $H$ where $P_1$ is just a single vertex that lies in $K$.
    Then either we can find a new mandatory edge, or we can mark $|D|^3 \cdot (k+2)$ vertices in $K$ such that the following holds.
    For any set $X$ of $k$ vertices such that $G \setminus  X$ has a chordless cycle which intersects $K$ in an unmarked vertex, there is another chordless cycle in $G \setminus X$ 
    which does not use any unmarked vertices of $K$.
\end{lemma}
\begin{proof}
    Our proof of this lemma is obtained by modifying proof of Lemma 18,~\cite{Marx10}.
    And note that, as $P_1$ is a single vertex, $F$ must have at least two vertices.
    For every $l_1, l_2, l_3 \in D$, mark $k+1$ vertices of $K$ which have labels $l_1, l_2$ and not $l_3$. Hence, in total we have marked $|D|^3 \cdot (k+1)$ vertices of $K$.
    Marx~\cite{Marx10} shows that this is sufficient for the case when $F$ has $3$ or more vertices.
    
    When $F$ contains only two vertices $l_1, l_2$, we need to mark some additional vertices.
    Let us recall that $P_1$ is only a single vertex of $K$ and has the labels $l_1, l_2$,
    and so it follows that $l_1,l_2$ are non-neighbors.
    Let $x$ be the bag in the clique-tree of $G \setminus  D$, which corresponds to the maximal clique $K$, and further assume that $x$ is the root of this tree-decomposition.
    For any chordless cycle $H^i$ such that $H^i \setminus D = P^i_1, P^i_2$ where $P^i_1$ is a single vertex, let $w_i$ be the bag closest to $x$ in the tree decomposition that contains a vertex of $P^i_2$. 
    Note that $w_i \neq x$, as vertices of $P^i_1$ and $P^i_2$ have no edges between them, and it follows that the vertices of $P^i_2$ are contained in the sub-clique-tree rooted at $w_i$,
    and furthermore none of them are present in any bag outside this sub-clique-tree.
    We can generate a list of these bags $w_1, w_2, \ldots$, by considering every choice of $P^i_1$, $l_1$ and $l_2$ and selecting those bags of the tree decomposition that have a path from $l_1$ to $l_2$ with at least one internal vertex in the non-neighborhood of the choice of $P^i_1$.
    Among these bags, select $w_{i_1}, w_{i_2}, \ldots, w_{i_q}$ such that none of its descendants in the clique-tree are in the collection $w_1, w_2, \ldots$ computed above.
    Note that, by definition no bag selected above is an ancestor or a descendant of another selected bag.
    Let $H_{i_1}, H_{i_2}, \ldots, H_{i_q}$ be a collection of chordless cycles such that $P^{i_j}_2$ is contained in the sub-clique-tree of $w_{i_j}$, for $j=1,2 \ldots, q$.
    Note that the collection of $P^{i_j}_2$ are pairwise vertex disjoint, and further there are no edges between the vertices of $P^{i_j}_2$ and $P^{i_{j'}}_2$ for $j \neq j'$.
    This follows from the fact that no bag outside the sub-clique-tree rooted at $w_{i_j}$
    contains any vertex of $P^{i_j}_2$, and that no bag selected in the collection of $w_{i_j}$ is an ancestor or a descendant of another.
    
    Consider the case when $q \leq k+1$.
    We define the distance of a vertex $v$ from a bag $w$ in the clique-tree as the minimum of the distance between $w$ and a bag $x$ that contains $v$.
    Then for each $w_{i_j}$, sort the vertices of of $K$, which have labels $l_1, l_2$, according to the distance from $w_{i_j}$, and mark $k+1$ vertices from $w_{i_j}$.
    It follows that we mark $|D|^2 \cdot (k+1)^2$ vertices of $K$.
    Let us argue that these marked vertices satisfy the requirements.
    For any set $X$ of at most $k$ vertices, suppose $H$ is a chordless cycle in $G \setminus X$ with fragments $F, P_1, P_2$ where $P_1$ is just a single vertex $u \in K$ with labels $l_1, l_2 \in F$.
    Consider some bag $w_i$ whose sub-clique-tree contains the path $P_2$, and note that $w_i$, or some descendent $w_{i_j}$ was selected in the above collection.
    Now if $u$ was not marked for $w_{i_j}, l_1, l_2$, then any of the vertices that were marked for this tuple has a greater distance from the bags $w_{i_j}$ (and $w_i$) than $u$.
    Since we marked $k+1$ vertices for this tuple, at least one of them is present in $K \setminus X$ and let $u'$ be that vertex.
    It follows that replacing $u$ with $u'$ in $H$ gives us a chordless cycle in $G \setminus  X$.
    
    Now, consider the case when $q \geq k+2$.
    The collection of paths $P^{i_j}$, along with $l_1$ and $l_2$, defines a collection of chordless cycles such that any solution of size $k$ must pick at least one of $l_1$ or $l_2$.
    Since $l_1,l_2$ are non-neighbors, we obtain a new mandatory edge $(l_1,l_2)$.
    
    Further note that the total number of marked vertices is $|D|^3 \cdot (k+1) + |D|^2 \cdot (k+1)^2$.
    Since $|D| \geq k$, it follows that the total number of marked vertices is upper bounded by $|D|^3 \cdot (k+2)$.
\end{proof}

Now, the only remaining case is when the path $P_1$ has two or more vertices. For this case, we have the following lemma, which follows from Lemma 19 of \cite{Marx10} where we rely on Lemma \ref{lem:cliMark1}.

\begin{lemma}[Lemma 19,~\cite{Marx10}]\label{lem:cliMark3}
    Let $F, P_1, \ldots, P_s$ be the fragments of a chordless cycle $H$ where $P_1$ contains at least two vertices, and further at least one of those vertices are in $K$.
    Then we can mark at most $|D| \cdot 6k^2(k+1)+|D| \cdot \Delta \cdot (k+2)^3$ vertices in $K$ such that, for any set $X$ of at most $k$ vertices, if $G \setminus X$ contains a chordless cycle that includes an unmarked vertex $u \in K$, then $G \setminus X$ also contains a chordless cycle that avoids $u$, and it has strictly fewer unmarked vertices of $K$.
\end{lemma}

Combining all of the above lemmas we obtain the following.

\begin{lemma}\label{lem:boundMaxCli-2}
    Let $(G,k)$ be an instance of \cdel whose independent-degree is bounded by $\Delta$.
    Let $D$ be solution to this instance and let $K$ be a maximal clique in $G \setminus D$.
    Then in polynomial time, either we can find a vertex which must be part of any solution of size at most $k$, or we can find a new mandatory edge, or we can safely remove all but $c\cdot(|D|^3 \cdot k + |D| \cdot \Delta \cdot (k+2)^3)$ vertices of $K$. (Here $c$ is some constant independent of the input.)
\end{lemma}

\begin{proof}
    We apply the Lemmas~\ref{lem:cliMark1},~\ref{lem:cliMark2} and ~\ref{lem:cliMark3} to the given instance.
    If any of them return a vertex that must be part of any solution of size at most $k$, or a new mandatory edge then we output that vertex or edge.
    Otherwise, together they mark a maximum of $c' \cdot(|D|^3 \cdot k + |D|\cdot \Delta \cdot (k+2)^3)$ vertices in $K$, where $c'$ is some constant independent of the input.
    In addition, let us also mark any unmarked vertex in $K$ which is an endpoint of a mandatory edge in the instance.
    Since the total number of mandatory edges is always bounded by $k^2$, the total number of marked vertices in $K$ does not exceed $c \cdot (|D|^3 \cdot k + |D|\cdot \Delta \cdot (k+2)^3)$,
    where $c$ is again a constant independent of the input.
    
    Let $u$ be any unmarked vertex in $K$.
    We will show that the instances $G$ and $G \setminus u$ are equivalent.
    Let $X$ be any set of at most $k$ vertices such that  $(G \setminus u) \setminus X$ is a chordal graph.
    If $G \setminus  X$ is not chordal, then there is some chordless cycle $H$ in $G \setminus X$.
    If $H$ does not contain the vertex $u$, then it is also present in $(G \setminus u) \setminus X$ which is a contradiction.
    On the other hand if $H$ contains $u$, then by the above lemmas, we can argue that
    there is some chordless cycle $H'$ in $G \setminus X$ which avoids $u$ (see Lemma 19,~\cite{Marx10}).
    Therefore $H'$ is present in $(G \setminus u) \setminus X$, which is again a contradiction.
    Now, by induction, we can show that it is safe to remove all the unmarked vertices in $K$ from the graph.
\end{proof}

Now we are ready to prove Lemma~\ref{lem:boundMaxCli}.
Consider any maximal clique in the graph $G \setminus  D$ with more than $c\cdot(|D|^3 \cdot k + |D| \cdot \Delta \cdot (k+2)^3)$ vertices, and apply Lemma~\ref{lem:boundMaxCli-2} to it.
If it returns a vertex $v$ that must be part of any solution of size $k$, we remove it from the graph and decrease $k$ by one.
Else if it returns a mandatory edge, we add this edge to the graph and mark it as mandatory.
We can argue, as before, that both these operations are safe.
Otherwise Lemma~\ref{lem:boundMaxCli-2} ensures that it is safe remove all but the $c\cdot(|D|^3 \cdot k + |D| \cdot \Delta \cdot (k+2)^3)$ marked vertices of $K$.
Therefore, we remove these vertices from the graph.
Observe that, each application of Lemma~\ref{lem:boundMaxCli-2} either reduces $k$, or adds a new mandatory edge, or reduces the number of vertices in the graph.
Further, we may add at most $k^2 + 1$ new mandatory edges before finding a new vertex that must be part of any solution of size $k$ or concluding that the given instance is a No instance.
Hence, in polynomial time, either we bound the size of every maximal clique in the graph
or conclude that the given instance is a No instance.

\end{document}